\documentclass{sig-alternate}

\usepackage{graphicx}
\usepackage{caption,subcaption}
\DeclareCaptionType{copyrightbox}
\usepackage{url}
\usepackage{balance}

\begin{document}

\newtheorem{theorem}{Theorem}
\newtheorem{experiment}{Experiment}

%

\title{A Tighter Real-Time Communication Analysis for Wormhole-Switched Priority-Preemptive NoCs}

\numberofauthors{3}

\author{
%
%
\alignauthor
Borislav Nikoli\'{c}\\
       \affaddr{CISTER/INESC TEC, ISEP}\\
       \affaddr{Polytechnic Institute of Porto}\\
       \affaddr{Portugal}\\
       \email{borni@isep.ipp.pt}
\alignauthor
Leandro Soares Indrusiak\\
       \affaddr{Real-Time Systems Group}\\
       \affaddr{Dept. of Computer Science}\\
       \affaddr{University of York, UK}\\
       \email{lsi@cs.york.ac.uk}
\alignauthor
Stefan M. Petters\\
       \affaddr{CISTER/INESC TEC, ISEP}\\
       \affaddr{Polytechnic Institute of Porto}\\
       \affaddr{Portugal}\\
       \email{smp@isep.ipp.pt}
}
\date{27 June 2014}

\maketitle
\begin{abstract}

	Simulations and runtime measurements are some of the methods which can be used to evaluate whether a given NoC-based platform can accommodate application workload and fulfil its timing requirements. Yet, these techniques are often time-consuming, and hence can evaluate only a limited set of scenarios. Therefore, these approaches are not suitable for safety-critical and hard real-time systems, where one of the fundamental requirements is to provide strong guarantees that all timing requirements will always be met, even in the worst-case conditions. For such systems the analytic-based real-time analysis is the only viable approach.
	
	In this paper the focus is on the real-time communication analysis for wormhole-switched priority-preemptive NoCs. First, we elaborate on the existing analysis and identify one source of pessimism. Then, we propose an extension to the analysis, which efficiently overcomes this limitation, and allows for a less pessimistic analysis. Finally, through a comprehensive experimental evaluation, we compare the newly proposed approach against the existing one, and also observe how the trends change with different traffic parameters.
	
\end{abstract}

\section{Introduction}

	The technological advancements in the semiconductor technology have led to a stage where further processing power enhancements related to single-core platforms are no longer affordable. Consequently, chip manufacturers took a design paradigm shift and started integrating multiple cores within a single chip~\cite{SCC,Tilera}. Nowadays, platforms consisting of several cores ({\em multi-cores}) and more than a dozen of cores ({\em many-cores}) are commonplace in many scientific areas, e.g. high-performance computing, while they are still an emerging technology in others, like safety critical and real-time systems.
	
	The Network-on-Chip (NoC)~\cite{Benini_DeMicheli_02} architecture became the prevailing interconnect medium in many-cores, due to its scalability potential~\cite{Kavaldjiev_Smit_03}. For NoCs, one of the most popular data transfer techniques is the {\em wormhole switching}~\cite{Ni_McKinley_93}, because of its good throughput and small buffering requirements. Currently available wormhole-switched NoCs employ a wide range of diverse strategies when transferring the data, e.g. different sizes of basic transferable units -- {\em flits}, different router operating frequencies, the (in)existence of virtual channels, different arbitration policies. These design choices have a significant impact on both the performance and the analysis.
	
	The behaviour and performance of NoC-based platforms can be evaluated via simulations and/or runtime measurements. These evaluation techniques are suitable options for some areas, e.g. general-purpose and high-performance computing. However, in some other areas, like the safety-critical and real-time computing, the most important aspect is {\em not} the overall system performance, but instead, whether a platform can fulfil all timing requirements of a given application workload, even in the worst-case conditions. Simulations and runtime measurements are not suitable approaches for these domains, because they are usually time consuming, and hence can evaluate only a small fraction of all possible scenarios, without any information whether the worst-case scenario was indeed captured, or not. In such cases, the analytic-based real-time analysis is the only viable approach.
	
	Yet, analysis-based approaches are mostly performed at design-time. Thus, their efficiency highly depends on the amount of predictability of the entire system, whereas any non-deterministic system behaviour has to be accounted for in the analysis with a certain degree of pessimism. A more pessimistic analysis may lead to a significant resource overprovisioning and/or underutilisation of platform resources. Conversely, a less pessimistic approach allows to save on design costs, e.g. by choosing a cheaper platform with fewer resources, which still guarantees the fulfilment of all timing constraints. Moreover, with the less pessimistic analysis, the resources of the platform can be exploited more efficiently, for example, by accommodating the additional workload, or by decreasing the power consumption via core shutdowns and smaller router frequencies. Thus, deriving the analysis with as less pessimism as possible is of paramount importance in the real-time domain.
	
	{\bf Contribution:} In this paper the focus is on the real-time communication analysis for wormhole-switched priority-pre\-emptive NoCs. Specifically, the main objective is to derive guarantees that all traffic flows complete their traversal over the NoC without violating posed timing constraints, even in the worst-case conditions. First, we elaborate on the existing analysis~\cite{Zheng_Burns_08} ({\bf Section~\ref{sec:background}}), and identify one source of pessimism ({\bf Section~\ref{sec:pessimism}}). Then, we propose the extension, which overcomes the aforementioned limitation, and allows for a less pessimistic analysis ({\bf Section~\ref{sec:approach}}). Finally, we compare the new approach with the existing one, and observe how the trends change with different traffic parameters ({\bf Section~\ref{sec:observations}, Section~\ref{sec:example} and Section~\ref{sec:experiments}}).
	
\section{Related Work}
\label{sec:related_work}

	Contrary to the popular belief, the wormhole switching technique~\cite{Ni_McKinley_93} is not a novelty. In fact, it has been introduced more than $ 20 $ years ago. However, it has been largely neglected, since the alternative {\em store-and-forward switching} technique was providing satisfactory results~\cite{Kavaldjiev_Smit_03}. Yet, as the data that has to be transferred kept increasing, buffering within routers became a challenge, which lately brought the wormhole switching back into focus. Nowadays, some NoC-based many-cores employ this technique, e.g.~\cite{SCC,Tilera}.
	
	When organising the access to the interconnect medium, NoC-based platforms employ a diverse range strategies. All approaches can be broadly classified into two categories: {\em contention-free} techniques and {\em contention-aware} techniques. An example of the former category is the AEthereal platform~\cite{Goossens_DR_05}, where a time-division-multiple-access approach is used to organise the access to interconnect resources. Conversely, if contentions among traffic flows are allowed, the performance of the system may improve, but the analysis of the system behaviour becomes more complex as well. This is especially emphasized in cases where a platform provides only a single virtual channel. For this model, several analyses have been proposed to compute the worst-case traversal times of traffic flows~\cite{Dasari_NNP_13, Ferrandiz_FF_09, Ferrandiz_FF_11, Qian_LD_09}, however, due to complex interference patterns, the obtained values may be overly pessimistic~\cite{Dasari_NNP_13}.
	
	If a platform provides multiple virtual channels~\cite{Dally_92,Dally_Seitz_87}, the benefits are twofold: (i) the throughput significantly increases~\cite{Dally_92,Dally_Seitz_87}, and (ii) traffic preemptions can be implemented~\cite{Song_KY_97}. Shi and Burns~\cite{Zheng_Burns_08} proposed the real-time analysis assuming: per-flow distinctive priorities, per-priority virtual channels and flit-level preemptions. Subsequently, Nikoli\'{c} et al.~\cite{Nikolic_APP_13} relaxed the requirements for virtual channels, proving that the analysis can remain unaffected as long as the number of virtual channels is at least equal to the maximum number of contentions for any port of any router, which is a realistic assumption. Indeed, assuming a 2D mesh NoC interconnect with a $ 10 \times 10 $ grid and $ 400 $ traffic flows, on average, only $ 8 $ virtual channels are needed~\cite{Nikolic_APP_13}, while platforms with $ 8 $ virtual channels (e.g.~\cite{SCC}) were already available more than $ 5 $ years ago.
	
	So far, the aforementioned model (a wormhole-switched priority-aware NoC with flit-level preemptions via virtual channels) appears to be a promising steps towards real-time NoCs. Therefore, providing any improvement over the state-of-the-art analysis~\cite{Zheng_Burns_08} would be a valuable contribution, primarily to the safety-critical and the real-time domain. Our work is motivated with this reasoning.
	
\section{System Model}
\label{sec:model}

\subsection{Platform}

	\begin{figure}
		\centering
		\includegraphics[width=0.9\columnwidth]{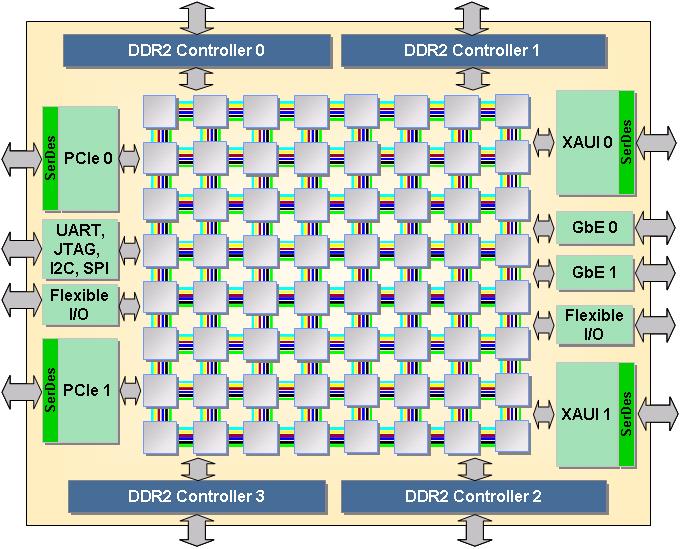}
		\caption{\label{f:platform} TILEPro64 Platform}
	\end{figure}

	We consider a many-core platform $ \Pi $, comprised of $ m \times n $ tiles, interconnected with a 2D mesh NoC, like the system depicted in Figure~\ref{f:platform}. Each tile contains a single core and a single router. Each router has a set of ports, which are used to exchange the data with the local core, as well as with the neighbouring routers. Each pair of communicating ports is connected with two unidirectional links. Additionally, the platform employs a wormhole switching technique with the credit-based flow-control mechanism, where acknowledgements use separate physical links (see the upper part of Figure~\ref{f:router}). This means that, prior to sending, each data packet is divided into small fixed-size elements called {\em flits}. The header flit establishes the path, and the rest of the flits follow in a pipeline manner. Moreover, we assume a deterministic, dimension-ordered XY routing mechanism, which is deadlock and livelock free~\cite{Hu_Marculescu_03}.
	
	The platform provides virtual channels. A virtual channel is implemented as an additional buffer within every port of every router. Virtual channels are used as an infrastructure to implement the priority-preemptive router-arbitration policy, i.e. to store flits of blocked/preempted packets and offer the network resources to other packets which can freely progress. This significantly improves the throughput (performance)~\cite{Dally_92,Dally_Seitz_87}, and also makes the system behaviour more predictable and deterministic, which is of paramount importance for the real-time analysis. The requirement is that the number of virtual channels is at least equal to the maximum number of contentions for any port of any router. As already described (Section~\ref{sec:related_work}), this guarantees that flits of each packet will have an available virtual channel in every port along the path~\cite{Nikolic_APP_13}. The router architecture is depicted in the upper part of Figure~\ref{f:router}. Notice, that our target platform is very similar to existing many-core systems (e.g.~\cite{SCC,Tilera}) and that the required hardware features already exist within these platforms. Note, our method also applies to any NoC topology and any deterministic routing technique, which guarantee a single continuous contention domain between any pair of contending traffic flows. For clarity purposes, the focus of this paper is on 2D mesh NoCs with the XY routing mechanism.
	
	\begin{figure}
		\centering
		\includegraphics[width=0.9\columnwidth]{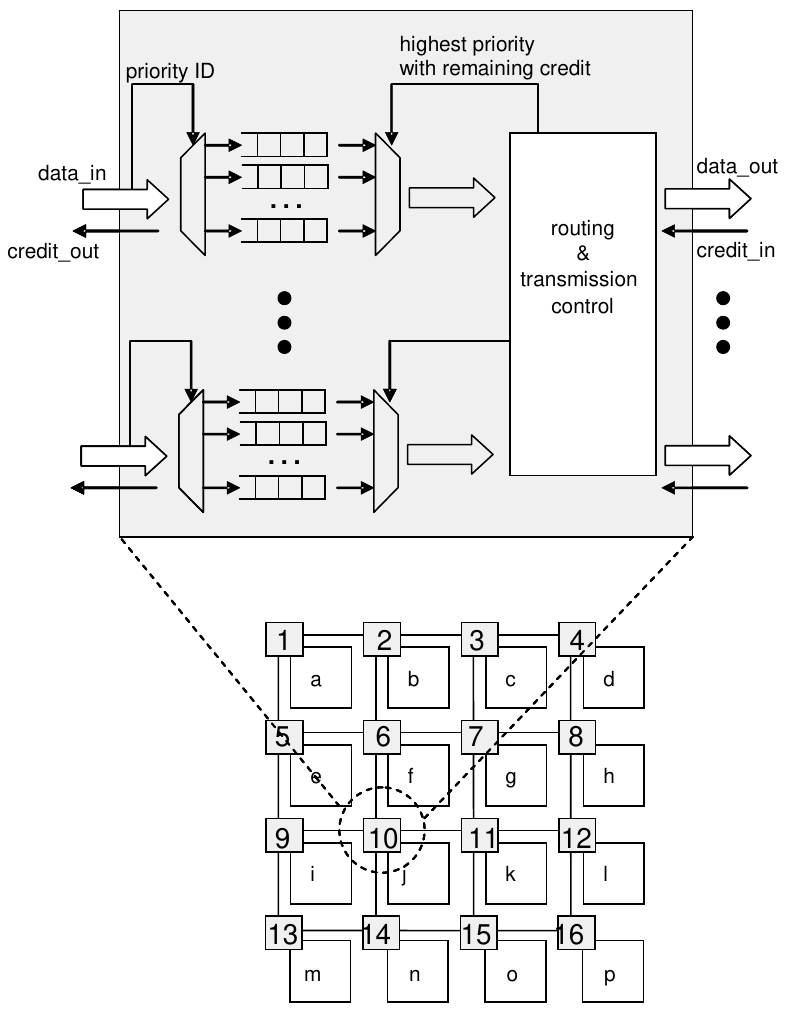}
		\caption{\label{f:router} Platform and router architecture}
	\end{figure}
	
\subsection{Workload}

	The workload consists of a traffic flow-set $ {\cal F} $, which is a collection of flows $ \{ f_1, f_2, ... f_{z-1}, f_z \} $. Each flow $ f_i \in {\cal F} $ is characterised by a deterministic XY-routed path from its source core to its destination core, expressed as a set of traversed links $ {\cal L}_i $, where $ |{\cal L}_i| $ denotes the cardinality of the path, which is, in the literature, often called {\em the number of hops}. Additional flow parameters are: its size $ size(f_i) $, its unique priority $ P_i $, its minimum inter-release period $ T_i $, and its deadline $ D_i $. In this paper we assume that all flows have implicit deadlines, i.e. $ \forall f_i \in {\cal F} : D_i = T_i $. Each flow $ f_i $ generates a potentially infinite sequence of packets. A packet released at the time instant $ t $ should be received by the destination no later than $ t + D_i $. If it fails to do so, it has {\em missed a deadline}. If guarantees can be provided that a flow will not miss its deadline, even in the worst-case conditions, it is considered {\em schedulable}. If all flows of the flow-set are schedulable, the flow-set is considered {\em schedulable}.
	
\section{Background and Preliminaries}
\label{sec:background}

	In this section we will describe the fundamental concepts and properties of the real-time communication analysis for NoCs, which are essential for understanding our approach.
	
	{\bf Prop.~1: }A traversal time of a packet, belonging to a flow $ f_i $, termed $ C_i $ (Equation~\ref{eq:iso_lat}), is equal to the time it takes for its header flit to establish the path and reach the destination, augmented by the transfer time of the rest of the flits, once the header reaches the destination. $ d_r $ is the delay to route a packet header, $ d_l $ is the latency to transfer one flit across one link, while $ size(flit) $ represents the size of one flit. $ C_i $ is in the literature also known as the {\em basic network latency}, and it covers the case when the packet does not suffer any interference from the packets of other higher-priority flows.
	
	\begin{equation}
		\label{eq:iso_lat}
		\footnotesize
		C_i =  \overbrace{|{\cal L}_i| \times d_l + \left(|{\cal L}_i| - 1\right) \times d_r}^{\mbox{\footnotesize header to reach the destination}} + \overbrace{\left\lceil \frac{size(f_i)}{size(flit)} \right\rceil \times d_l}^{\mbox{\footnotesize the rest of the flits}}
	\end{equation}
	
	{\bf Prop.~2: }A flow $ f_i $ can be preempted only by higher-priority flows which share a part of the path with it, called {\em directly interfering flows}. Let $ {\cal F}_D(f_i) $ be a set of directly interfering flows of $ f_i $. Formally:
	
	\begin{equation}
		\label{eq:dir_inf}
		\footnotesize
		\forall f_j \in {\cal F} : \left( P_j > P_i \wedge {\cal L}_j \cap {\cal L}_i \not= \emptyset \right) \Rightarrow f_j \in {\cal F}_{D}(f_i)
	\end{equation}
	
	{\bf Prop.~3: }Interference caused by a single preemption of any $ f_j \in {\cal F}_D(f_i) $ to $ f_i $ is equal to its traversal time $ C_j $.
	
	{\bf Prop.~4: }The worst-case scenario for the flow $ f_i $ occurs when, during its traversal, all its directly interfering flows release their packets at their respective maximum possible rates. The traversal time of $ f_i $ in this scenario is called {\em the worst-case traversal time}, $ R_i $ in Equation~\ref{eq:response_old}. $ J_j^R $ is the release jitter~\cite{Audsley_BRTW_93}, and it denotes the maximum deviation of two successive packets released from their period $ T_j $. $ J_j^I $ is the interference jitter, and it symbolises the interference caused to $ f_i $ by other flows which are not directly interfering with it, but are directly interfering with $ f_j $ and hence can indirectly cause additional interference to $ f_i $. An interested reader can find more details about Equation~\ref{eq:response_old} in the previous work of Shi and Burns~\cite{Zheng_Burns_08}. Notice that Equation~\ref{eq:response_old} has a recursive notion, i.e. the term $ R_i $ exists on both sides of the equation. Therefore, it is solved iteratively~\cite{Audsley_BRTW_93}, until reaching a fixed converging point. If it cannot be solved (e.g. in the case of starvation), the flow is deemed unschedulable. 
	
	\begin{equation}
		\label{eq:response_old}
		\footnotesize
		R_i = C_i + \sum\limits_{\forall f_j \in {\cal F}_D(f_i)} \left\lceil \frac{R_i + J_j^R + J_j^I}{T_j} \right\rceil \times C_j
	\end{equation}
	\vspace{2ex}
	
	{\bf Prop.~5: }If the computed worst-case traversal time of the flow $ f_i $ is less than or equal to its deadline, i.e. $ R_i \leq D_i $, this means that the packets of $ f_i $ will never miss the deadline due to the other traffic existing within the same platform, and thus $ f_i $ is schedulable. By computing the worst-case traversal times of all flows of the flow-set, it can be determined whether the entire flow-set is schedulable or not.
	
	Note, the value of $ R_i $, obtained by solving Equation~\ref{eq:response_old}, does not present the exact worst-case traversal time that $ f_i $ might experience while traversing, but rather its analytically computed safe upper-bound estimate (see Figure~\ref{f:exmpl_results}). In some cases, the upper-bounds obtained by solving Equation~\ref{eq:response_old} may be pessimistic, which is thoroughly discussed in the next section.
	
\section{Proposed Approach}

\begin{figure}
		\centering
		\includegraphics[width=0.9\columnwidth]{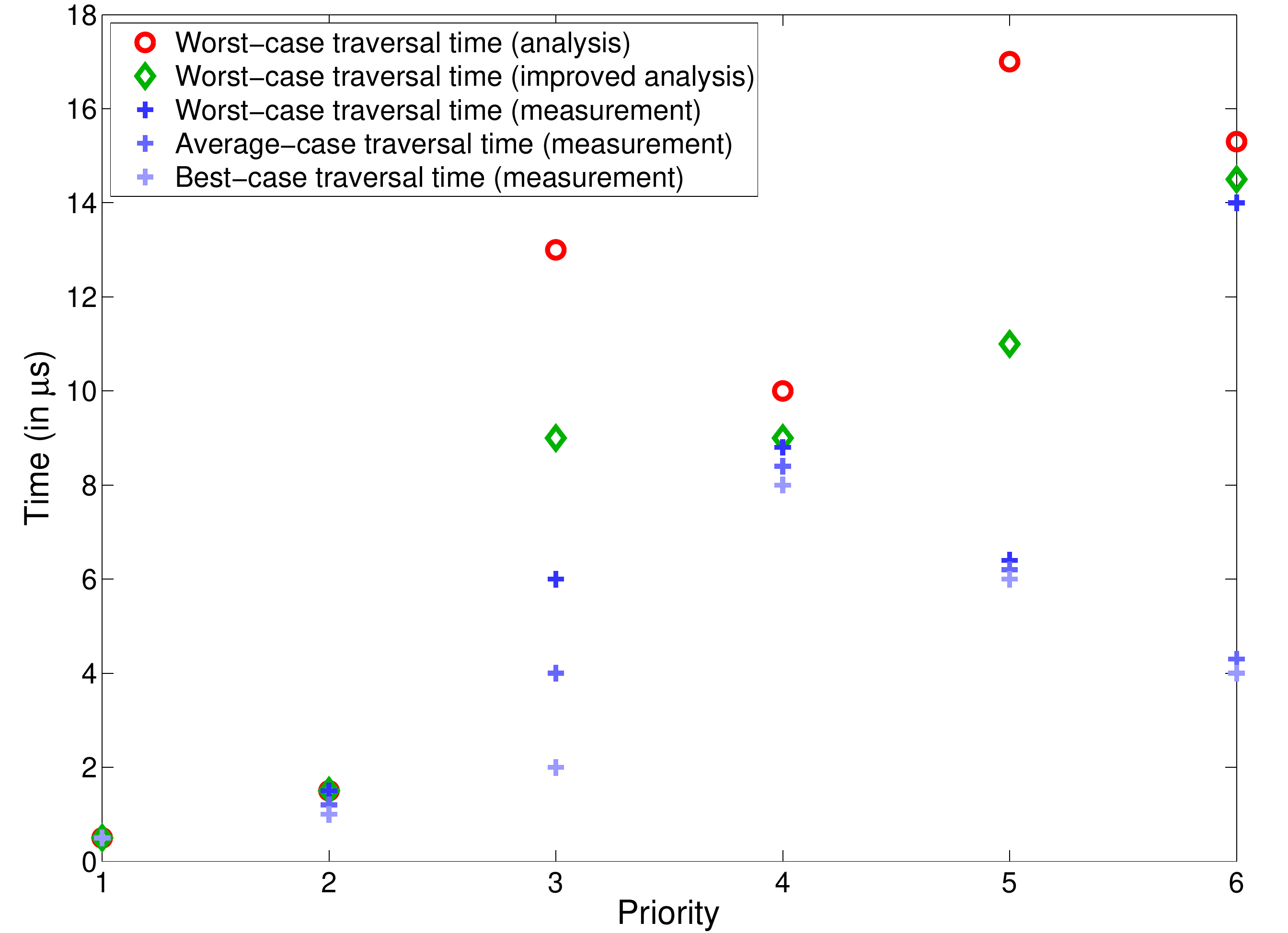}
		\caption{\label{f:exmpl_results} Example of flow traversal times}
	\end{figure}
	
\subsection{Work Objectives}
\label{sec:objectives}
	
	The work objectives can be summarised as follows. Given the platform $ \Pi $ and the flow-set $ {\cal F} $, provide the analysis to obtain the upper-bound estimates on the worst-case traversal times of individual flows, such that the obtained bounds are as tight as possible (with as less pessimism as possible).
	
	We illustrate this with the example given in Figure~\ref{f:exmpl_results}. Each flow has its traversal times obtained via (i) the analysis, and (ii) measurements. The goal is to reduce the gap between the analytic and the measured worst-case values, i.e. to provide the analysis which will produce significantly tighter (less pessimistic) estimates, e.g. green diamonds in Figure~\ref{f:exmpl_results}.

\subsection{Source of Pessimism in Existing Analysis}
\label{sec:pessimism}

	\begin{figure}[t!]
		\centering
		\includegraphics[width=0.9\columnwidth]{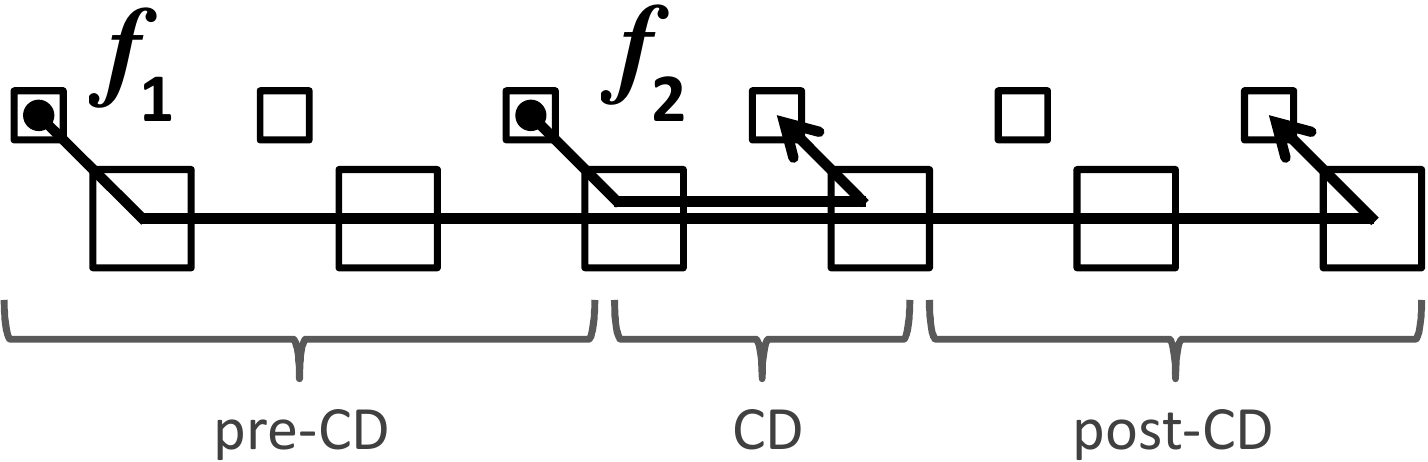}
  		\caption{\label{f:exmpl1} Example of traffic flows}
	\end{figure}

	In the existing analysis~\cite{Zheng_Burns_08} (Section~\ref{sec:background}), it is assumed that the flow under analysis suffers the interference from the entire traversal of a higher-priority flow (see Prop.~3). However, that assumption may be overly pessimistic, as described below.
	
	Consider the example of two flows, $ f_1 $ and $ f_2 $, illustrated in Figure~\ref{f:exmpl1}, where $ P_1 > P_2 $. Small and big rectangles depict cores and routers, respectively. Recall, that in the existing analysis, the entire traversal of $ f_1 $ would be considered as the interference to $ f_2 $ (Prop.~3). Let us closely investigate the paths of these flows. $ f_1 $ and $ f_2 $ share the common part of the path, which consists of one link, and which we hereafter refer to as {\em the contention domain} -- CD\footnote{The XY-routing mechanism assures that any two flows with a direct interference relationship have exactly one CD}. Let us divide the path of $ f_1 $ into $ 3 $ parts: (i) before $ f_1 $ and $ f_2 $ start sharing a common part of the path -- {\bf pre-CD}, (ii) while $ f_1 $ and $ f_2 $ share the common part of the path -- {\bf CD}, and (iii) after $ f_1 $ and $ f_2 $ stop sharing the common part of the path -- {\bf post-CD}.
	
	Notice that while the header flit of $ f_1 $ traverses pre-CD, $ f_1 $ does not cause any interference to $ f_2 $ (see Figure~\ref{f:exmpl2}). Thus, in the existing analysis, the traversal of the header flit of $ f_1 $ through pre-CD is unnecessarily considered as the interference that $ f_1 $ causes to $ f_2 $. If we exclude that delay from the interference that $ f_1 $ causes to $ f_2 $, we can decrease the pessimism of the analysis and obtain a tighter upper-bound estimate on the worst case traversal time of $ f_2 $.
	
	Once the header flit of $ f_1 $ reaches CD, $ f_2 $ starts suffering the interference (see Figure~\ref{f:exmpl3}). When the tail flit of $ f_1 $ leaves CD, $ f_2 $ stops suffering the interference from $ f_1 $, and may continue its progress (see Figure~\ref{f:exmpl4}). Thus, the traversal of the tail flit of $ f_1 $ through post-CD is also unnecessarily considered as the interference in the existing analysis. By excluding this delay from the interference that $ f_1 $ causes to $ f_2 $, we can further decrease the pessimism of the analysis and obtain even tighter upper-bound estimate on the worst-case traversal time of $ f_2 $.
	
	\begin{figure}[t!]
		\centering
  		\begin{minipage}{1.0\columnwidth}
  			\centering
  			\includegraphics[width=0.7\columnwidth]{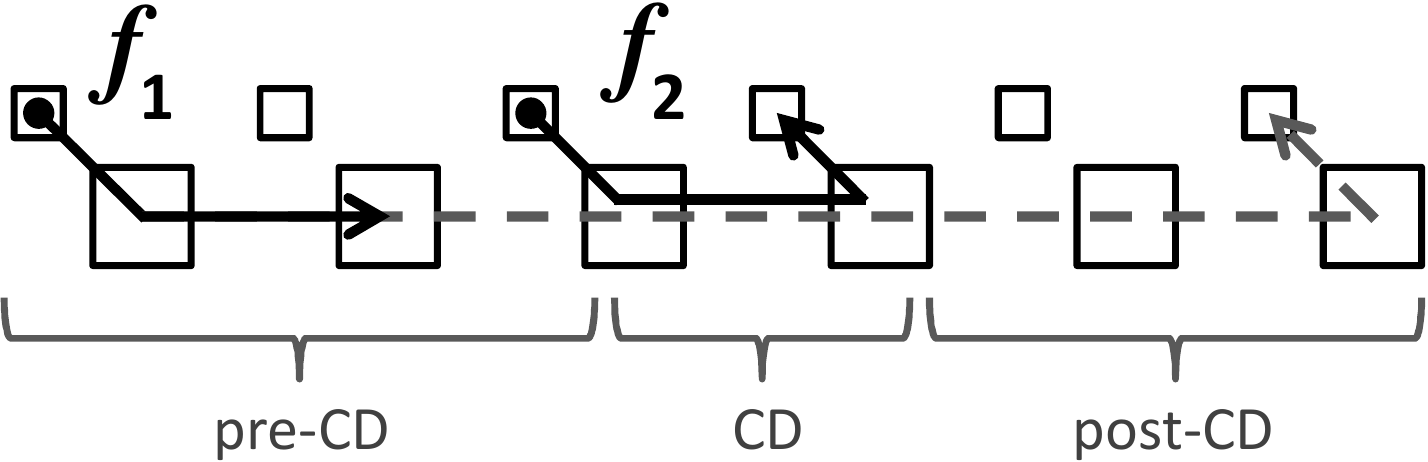}
  			\subcaption{\label{f:exmpl2} $ f_2 $ does not suffer the interference when the header flit of $ f_1 $ is in pre-CD}
  		\end{minipage}
  		\begin{minipage}{1.0\columnwidth}
  			\centering
			\includegraphics[width=0.7\columnwidth]{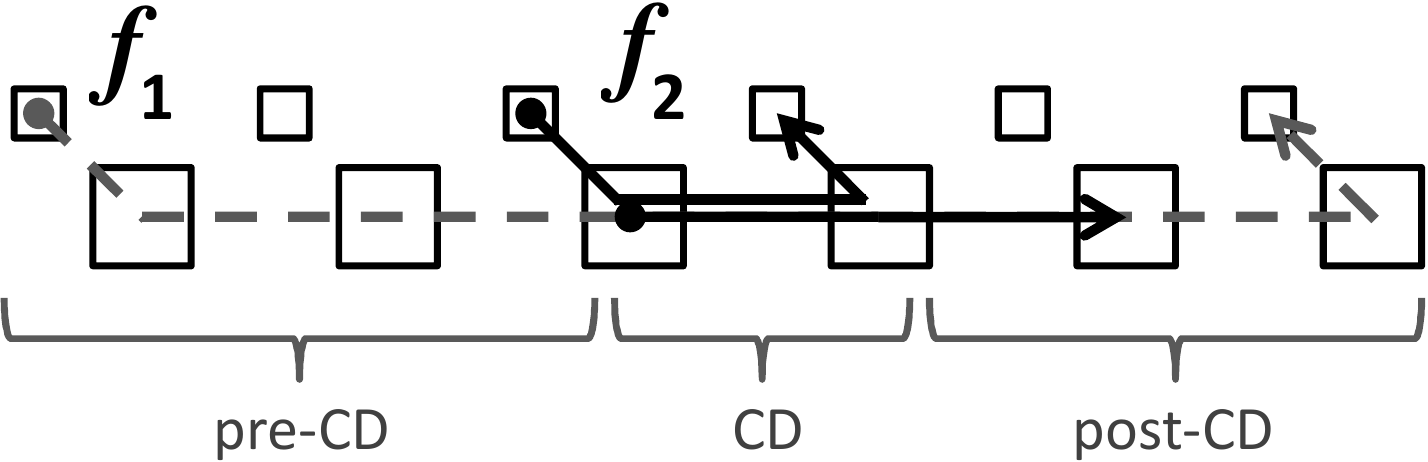}
  			\subcaption{\label{f:exmpl3} $ f_2 $ suffers the interference when the header flit of $ f_1 $ is in either CD or post-CD, and the tail flit of $ f_1 $ is in either pre-CD or CD}
  		\end{minipage}
  		\begin{minipage}{1.0\columnwidth}
  			\centering
			\includegraphics[width=0.7\columnwidth]{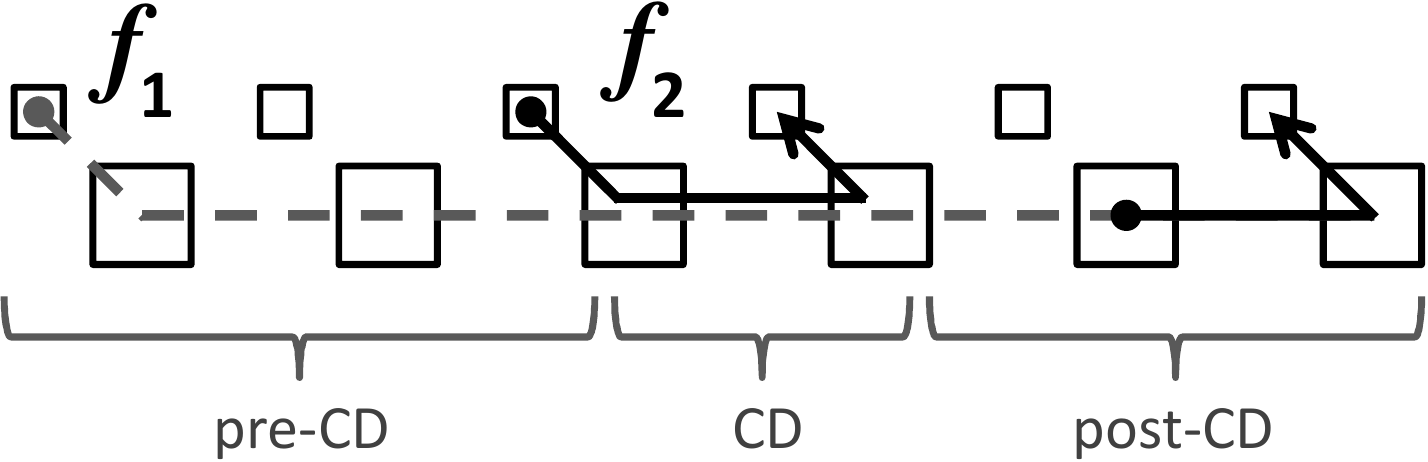}
  			\subcaption{\label{f:exmpl4} $ f_2 $ does not suffer the interference when the tail flit of $ f_1 $ is in post-CD}
  		\end{minipage}
  		\caption{Detailed analysis of the interference that $ f_1 $ causes to $ f_2 $}
	\end{figure}
	
\subsection{Less Pessimistic Worst-Case Analysis}
\label{sec:approach}

	In this section we propose an extension to the existing analysis. Specifically, based on the observations from the previous section, we develop a new, less pessimistic method to compute the interference that directly interfering flows cause to the flow under analysis. Subsequently, we can obtain tighter upper-bound estimates on the worst-case traversal times.
	
	Consider again the example illustrated in Figure~\ref{f:exmpl1}. Let the parts of $ {\cal L}_1 $ constituting sections pre-CD, CD and post-CD (with respect to flow $ f_2 $) be $ {\cal L}_{1,2}^{pre-CD} $, $ {\cal L}_{1,2}^{CD} $ and $ {\cal L}_{1,2}^{post-CD} $, respectively. It is obvious that the union of these parts forms the entire path of $ f_1 $, i.e. $ {\cal L}_1 = {\cal L}_{1,2}^{pre-CD} \cup {\cal L}_{1,2}^{CD} \cup {\cal L}_{1,2}^{post-CD} $. Now, as described in the previous section, the less pessimistic estimate of the interference that a single traversal of $ f_1 $ causes to $ f_2 $, denoted as $ I_{1,2} $ (Equation~\ref{eq:new_inf}), can be computed by subtracting (i) the time it takes a header flit of $ f_1 $ to traverse pre-CD $ \sigma^{pre-CD}_{1,2} $, and (ii) the time it takes a tail flit of $ f_1 $ to traverse post-CD $ \sigma^{post-CD}_{1,2} $ , from the entire traversal time of $ f_1 $.
	
	\begin{equation}
		\label{eq:new_inf}
		\footnotesize
		I_{1,2} = C_1 - \sigma^{pre-CD}_{1,2} - \sigma^{post-CD}_{1,2}
	\end{equation}
	
	where $ \sigma^{pre-CD}_{1,2} $ and $ \sigma^{post-CD}_{1,2} $ are computed as follows:
	
	\begin{equation}
		\footnotesize
		\sigma^{pre-CD}_{1,2} = \overbrace{|{\cal L}_{1,2}^{pre-CD}| \times d_l + max\left\{0, \left(|{\cal L}_{1,2}^{pre-CD}| - 1\right)\right\} \times d_r}^{\mbox{\footnotesize header flit traversal through pre-CD}}
	\end{equation}
	
	\begin{equation}
		\footnotesize
		\sigma^{post-CD}_{1,2} = \overbrace{|{\cal L}_{1,2}^{post-CD}| \times d_l}^{\mbox{\footnotesize tail flit traversal through post-CD}}
	\end{equation}
	
	If we apply this reasoning when computing the interference that a flow under analysis $ f_i $ suffers from all directly interfering flows, we can obtain a tighter worst-case traversal time $ R_i^* $, expressed with Equation~\ref{eq:response_new}. Note, that similar to Equation~\ref{eq:response_old}, Equation~\ref{eq:response_new} also has a recursive notion, thus is solved iteratively.
	
	\begin{equation}
		\label{eq:response_new}
		\footnotesize
		R_i^* = C_i + \sum\limits_{\forall f_j \in {\cal F}_D(f_i)} \left\lceil \frac{R_i^* + J_j^R + J_j^I}{T_j} \right\rceil \times I_{j,i}
	\end{equation}
	
\subsection{Observations}
\label{sec:observations}
	
	When we compare the existing and the new method to compute the worst-case traversal times of individual flows, we can notice several interesting facts:
	
	{\bf Observation 1: }When compared with the existing one, the proposed method never underperforms. Indeed, the new approach always derives upper-bounds which are either the same, or less pessimistic. In other words, it always holds that $ \forall f_i \in {\cal F} : R^*_i \leq R_i $. We prove that with Theorem~\ref{t:tightness}.
	
	\begin{theorem}
		\label{t:tightness}
		The worst-case traversal time, of any flow of the flow-set, computed with the new method, is either equal to or less pessimistic than the one obtained with the existing method.
	\end{theorem}
	\begin{proof}
		Proven directly. Consider two flows $ f_i $ and $ f_j $, where $ P_j > P_i $ and ${\cal F}_D(f_i) = \{ f_j \} $. Let $ \sigma_{j,i} $ be the difference in the interference caused by a single preemption of $ f_j $ to $ f_i $, computed with the existing and the proposed method (Equation~\ref{eq:sigmai}).
		
		\footnotesize
		\begin{displaymath}
			\sigma_{j,i} = C_j - I_{j,i} = \sigma^{pre-CD}_{j,i} + \sigma^{post-CD}_{j,i} = |{\cal L}_{j,i}^{pre-CD}| \times d_l +
		\end{displaymath}
		\begin{equation}
			\label{eq:sigmai}
			max\left\{0, \left(|{\cal L}_{j,i}^{pre-CD}| - 1\right)\right\} \times d_r + |{\cal L}_{j,i}^{post-CD}| \times d_l
		\end{equation}
		\normalsize
		
		Since all the terms of Equation~\ref{eq:sigmai} are non-negative values, it follows that $ \sigma_{j,i} \geq 0 $. Let $ K_{j,i} $ be the number of preemptions that $ f_j $ can cause to $ f_i $ during the worst-case traversal time of $ f_i $. Also, let $ \sigma_i $ be the difference in the total interference caused to $ f_i $ by all its directly interfering flows, computed with the existing and the proposed method (Equation~\ref{eq:sigma}).
		
		\begin{equation}
			\label{eq:sigma}
			\footnotesize
			\sigma_i = \sum\limits_{\forall f_j \in{\cal F}_D(f_i)} \sigma_{j,i} \times K_{j,i}
		\end{equation}
		
		Since all the terms of Equation~\ref{eq:sigma} have non-negative values, it follows that $ \sigma_i \geq 0 $. Moreover, as the existing and the new method differ only in the way how the interference is computed, it follows that $ R_i - R^*_i = \sigma_i \geq 0 $. 
	\end{proof}
	
	Theorem~\ref{t:safety} provides a proof that the obtained upper-bounds are safe.
	
	\begin{theorem}
		\label{t:safety}
		The traversal time of any packet belonging to the flow $ f_i $ can not be greater than $ R^*_i $ (Equation~\ref{eq:response_new}), even in the worst-case conditions.
	\end{theorem}
	\begin{proof}
		Proven directly. Shi and Burns have proven that $ R_i $ (Equation~\ref{eq:response_old}) presents the upper-bound on the worst-case traversal time of the flow $ f_i $ (see Theorem~1 in~\cite{Zheng_Burns_08}). Thus, if we prove that there exists some discontinuous time interval $ \sigma_i $ which is a part of $ R_i $, and during which $ f_i $ does not progress, nor any of its interfering flows causes the interference to it, we will prove that $ R_i - \sigma_i $ is also a safe upper-bound on the worst-case traversal time of $ f_i $.
		
		Consider two traffic flows $ f_i $ and $ f_j $, where $ P_j > P_i $ and $ {\cal F}_D(f_i) = \{ f_j \} $. According to Prop.~3, the entire traversal of $ f_j $ is considered as the interference that $ f_j $ causes to $ f_i $, and hence entirely contributes to $ R_i $. Let $ \sigma_{j,i} $ (Equation~\ref{eq:sigmai}) be the sum of: (i) $ \sigma^{pre-CD}_{j,i} $, which is the interval when the header flit of $ f_j $ traverses pre-CD and (ii) $ \sigma^{post-CD}_{j,i} $, which is the interval when the tail flit of $ f_j $ traverses post-CD. By definition, both a necessary and sufficient condition for the contention between $ f_i $ and $ f_j $ is that both of them attempt to traverse the CD section at the same time. However, during $ \sigma^{pre-CD}_{j,i} $ and $ \sigma^{post-CD}_{j,i} $, the flow $ f_j $ does not traverse CD, hence $ f_i $ can safely progress. Thus, the maximum interference that one packet of $ f_j $ can cause to $ f_i $ has a safe upper-bound, which is $ C_j - \sigma_{j,i} $. Subsequently, if during the traversal of $ f_i $, packets of $ f_j $ can appear at most $ K_{j,i} $ times, then the safe upper-bound on the interference that $ f_j $ can cause to $ f_i $ is $ K_{j,i} \times (C_j - \sigma_{j,i}) $. By elevating this reasoning, we conclude that the worst-case traversal time of $ f_i $ has a safe upper-bound $ R^*_i $ (Equation~\ref{eq:correctness}).
		
		\begin{equation}
			\label{eq:correctness}
			\footnotesize
			R^*_i = R_i - \sum\limits_{\forall f_j \in {\cal F}_D(f_i)} K_{j,i} \times \sigma_{j,i} = R_i - \sigma_i
		\end{equation}
	\end{proof}
	
	Let us observe the improvements of the proposed method over the existing one on a small-scale example given in Figure~\ref{f:exmpl1}, where the flow characteristics are as follows: $ size(f_1) = size(f_2) = 48B $, $ P_1 > P_2 $, $ T_1 = D_1 = T_2 = D_2 = 1000ns $. The other analysis parameters are given in Table~\ref{tab:params} (Section~\ref{sec:experiments}). For clarity purposes assume that the release jitters are equal to zero, i.e. $ J_1^R = J_2^R = 0 $. Moreover, since only two flows exist, there are no indirect interferences, thus interference jitters are equal to zero~\cite{Zheng_Burns_08}, i.e. $ J_1^I = J_2^I = 0 $.
	
	$ f_1 $ is the higher-priority flow, so its worst-case traversal time will be the same with both methods:
	
	\footnotesize
	\begin{displaymath}
		R_1 = R^*_1 = C_1 = |{\cal L}_1| \times d_l + \left(|{\cal L}_1| - 1\right) \times d_r + \left\lceil \frac{size(f_1)}{size(flit)} \right\rceil \times d_l =
	\end{displaymath}
	\begin{displaymath}
		7 \times 0.5ns + 6 \times 1.5ns + \left\lceil \frac{48B}{16B} \right\rceil \times 0.5ns = 14ns
	\end{displaymath}
	\normalsize
	
	However, the worst-case traversal time of $ f_2 $ is different. Let us first obtain the value with the existing analysis. For that, we need to compute $ C_2 $.
	
	\footnotesize
	\begin{displaymath}
		C_2 = |{\cal L}_2| \times d_l + \left(|{\cal L}_2| - 1\right) \times d_r + \left\lceil \frac{size(f_2)}{size(flit)} \right\rceil \times d_l =
	\end{displaymath}
	\begin{displaymath}
		3 \times 0.5ns + 2 \times 1.5ns + \left\lceil \frac{48B}{16B} \right\rceil \times 0.5ns = 6ns
	\end{displaymath}
	\normalsize
	
	We compute $ R_2 $ as follows:
	
	\footnotesize
	\begin{displaymath}
		R_2 = C_2 +  \left\lceil \frac{R_2 + J_1^R + J_1^I}{T_1} \right\rceil \times C_1
	\end{displaymath}
	\begin{displaymath}
		R_2^0 = 6 + \left\lceil \frac{0 + 0 + 0}{1000ns} \right\rceil \times 14ns = 6ns
	\end{displaymath}
	\begin{displaymath}
		R_2^1 = 6 + \left\lceil \frac{6ns + 0 + 0}{1000ns} \right\rceil \times 14ns = 20ns
	\end{displaymath}
	\begin{displaymath}
		R_2^2 = 6 + \left\lceil \frac{20ns + 0 + 0}{1000ns} \right\rceil \times 14ns = 20ns
	\end{displaymath}
	\normalsize
	
	Now, let us compute the worst-case traversal time of $ f_2 $ with the new approach. To do so, we first have to compute $ I_{1,2} $, which is the interference that $ f_1 $ causes to $ f_2 $ with a single preemption. The lengths of the relevant sections are: $ |{\cal L}_{1,2}^{pre-CD}| = 3 $, $ |{\cal L}_{1,2}^{CD}| = 1 $ and $ |{\cal L}_{1,2}^{post-CD}| = 3 $ (see Figure~\ref{f:exmpl1}).
	
	\footnotesize
	\begin{displaymath}
		I_{1,2} = C_1 - |{\cal L}_{1,2}^{pre-CD}| \times d_l - max\left\{0, \left(|{\cal L}_{1,2}^{pre-CD}| - 1\right)\right\} \times d_r
	\end{displaymath}
	\begin{displaymath}
		- |{\cal L}_{1,2}^{post-CD}| \times d_l
	\end{displaymath}
	\begin{displaymath}
		I_{1,2} = 14ns - 3 \times 0.5ns - 2 \times 1.5ns - 3 \times 0.5ns = 8ns
	\end{displaymath}
	\normalsize

	We compute $ R^*_2 $ as follows:
	
	\footnotesize
	\begin{displaymath}
		R_2^* = C_2 + \left\lceil \frac{R_2^* + J_1^R + J_1^I}{T_1} \right\rceil \times I_{1,2}
	\end{displaymath}
	\begin{displaymath}
		R_2^0 = 6 + \left\lceil \frac{0 + 0 + 0}{1000ns} \right\rceil \times 8ns = 6ns
	\end{displaymath}
	\begin{displaymath}
		R_2^1 = 6 + \left\lceil \frac{6ns + 0 + 0}{1000ns} \right\rceil \times 8ns = 14ns
	\end{displaymath}
	\begin{displaymath}
		R_2^2 = 6 + \left\lceil \frac{14ns + 0 + 0}{1000ns} \right\rceil \times 8ns = 14ns
	\end{displaymath}
	\normalsize
	
	We can see that the worst-case traversal time of $ f_2 $ obtained with the new approach is $ 14ns $, while it was $ 20ns $ with the existing analysis, which is an improvement of $ 6ns $, or in relative terms, an improvement of $ 30\% $.
	
	{\bf Observation 2: }The improvements of the proposed approach over the existing one depend on the lengths of the pre-CD, CD and post-CD sections of interfering flows. Consider again the example from Figure~\ref{f:exmpl1}, where $ P_1 > P_2 $. Assuming that the path of $ f_1 $ is constant, i.e. $ |{\cal L}_1| = const $, from Equations~\ref{eq:sigmai}-\ref{eq:sigma} it straightforwardly follows that the improvements in the worst-case traversal time of $ f_2 $ are greater when both $ |{\cal L}^{pre-CD}_{1,2}| $ and $ |{\cal L}^{post-CD}_{1,2}| $ are bigger. Since, the interference relationship between $ f_1 $ and $ f_2 $ is possible if and only if a contention between these two flows exist, i.e. $ |{\cal L}^{CD}_{1,2}| \geq 1 $, it follows that the necessary condition for the maximum improvements is: $ |{\cal L}^{pre-CD}_{1,2}| + |{\cal L}^{post-CD}_{1,2}| = |{\cal L}_1| - 1 $.
	
	We demonstrate that with an illustrative example given in Figure~\ref{f:exmpl5}. Consider the same flow characteristics as in the previous example, i.e. $ size(f_1) = size(f_2) = 48B $, $ P_1 > P_2 $, $ T_1 = D_1 = T_2 = D_2 = 1000ns, J_1^R = J_2^R = 0 $. If we compute the worst-case traversal times with both methods we get the following results: $ R_1 = R^*_1 = 14ns $, $ R_2 = 24ns $ and $ R^*_2 = 20.5ns $. When compared with the previous example, it is visible that the improvements in the worst-case traversal time of $ f_2 $ dropped from $ 6ns $ to $ 3.5ns $, or expressed relatively, from $ 30\% $ to less than $ 15\% $.
	
	Note, a special case occurs when paths of interfering flows entirely overlap, i.e. $ |{\cal L}^{pre-CD}_{1,2}| = {\cal L}^{post-CD}_{1,2}| = 0 $, $ |{\cal L}^{CD}_{1,2}| = |{\cal L}_1| $. In such scenarios there are no improvements because both approaches return the same values.
	
	\begin{figure*}[t]
		\centering
  		\begin{minipage}{0.32\linewidth}
			\includegraphics[scale=0.3]{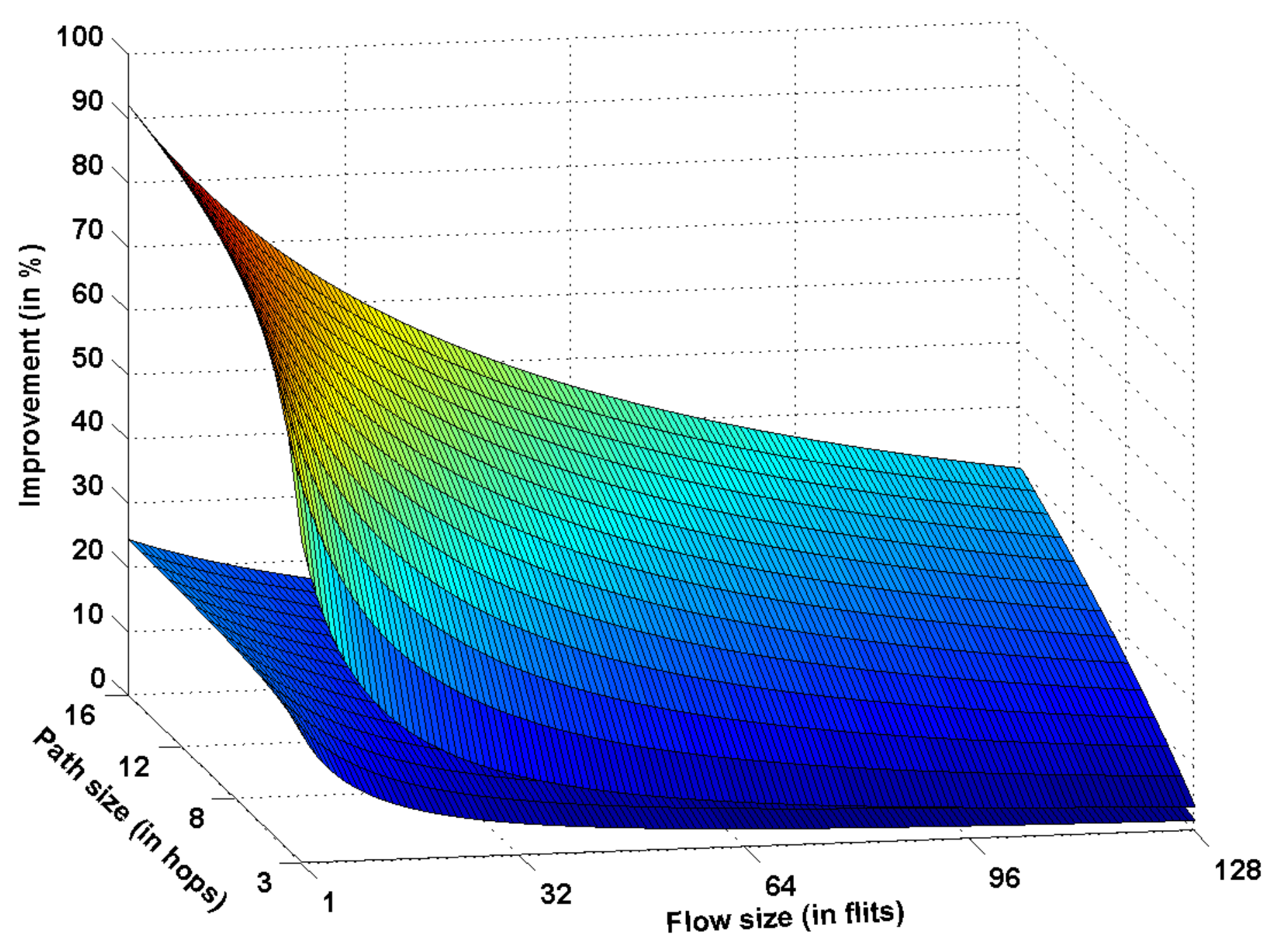}
  			\subcaption{\label{f:exp1a} $ |{\cal L}^{CD}_{1,2}| = 1 $ }
  		\end{minipage}
  		\begin{minipage}{0.32\linewidth}
			\includegraphics[scale=0.3]{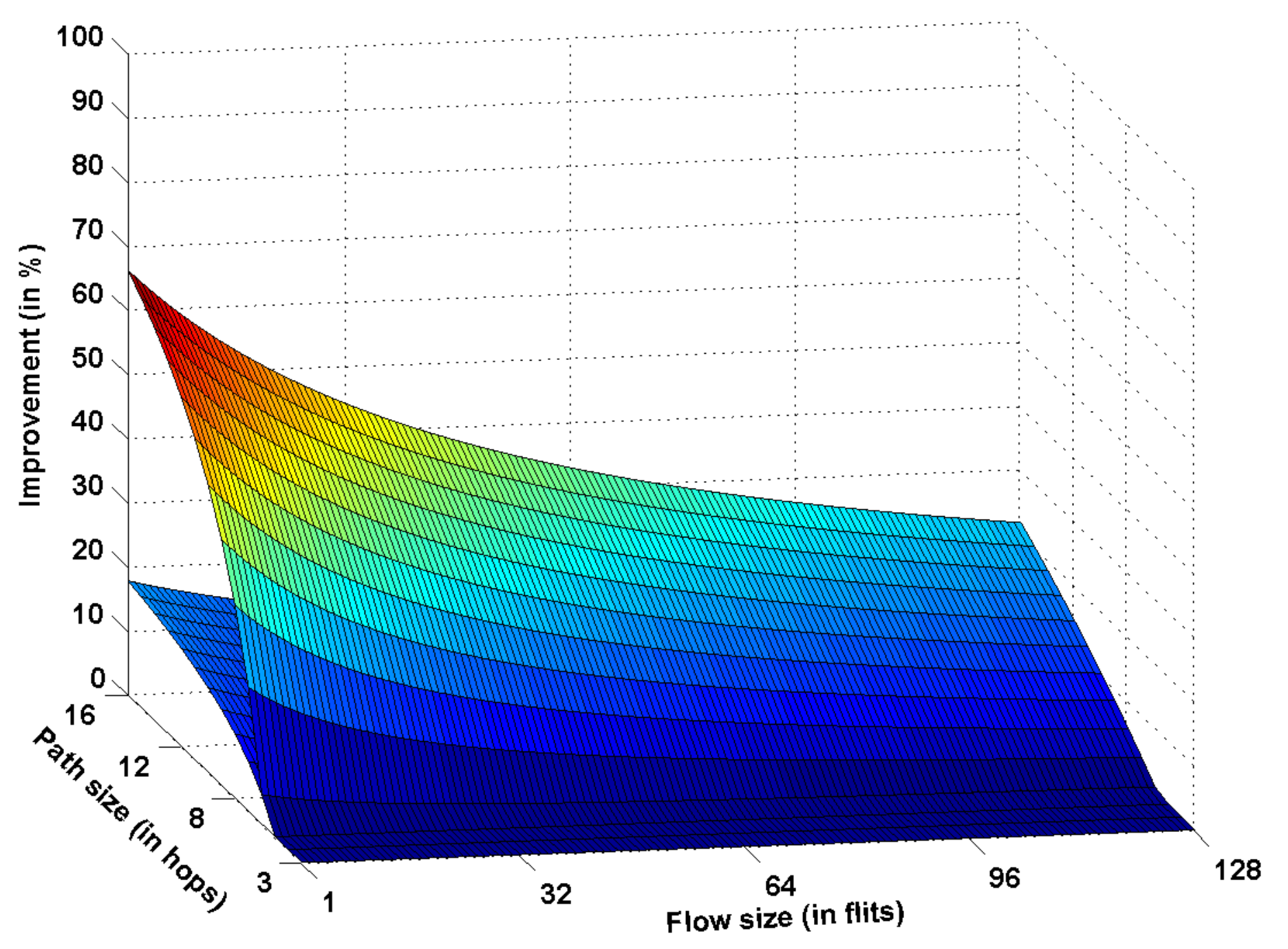}
  			\subcaption{\label{f:exp1b} $ |{\cal L}^{CD}_{1,2}| = min \{5, |{\cal L}_1| \} $}
  		\end{minipage}
  		\begin{minipage}{0.32\linewidth}
			\includegraphics[scale=0.3]{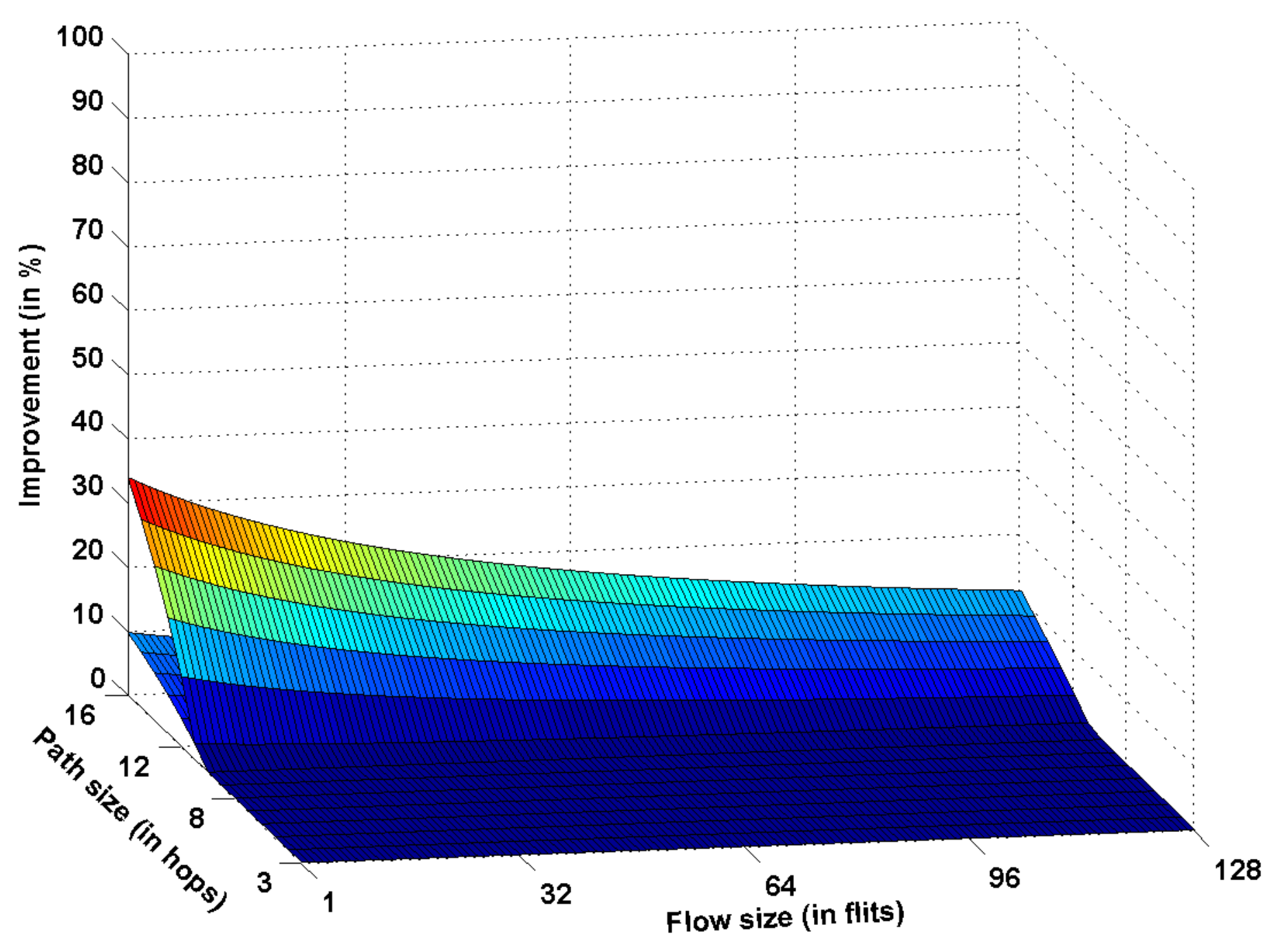}
  			\subcaption{\label{f:exp1c} $ |{\cal L}^{CD}_{1,2}| = min \{10, |{\cal L}_1| \} $}
  		\end{minipage}
  		\caption{\label{f:exp1} Numerical example with two contending flows $ f_1 $ and $ f_2 $, where $ P_1 > P_2 $ and $ |{\cal L}_1| = |{\cal L}_2| $}
	\end{figure*}
	
	{\bf Observation 3: }The improvements of the proposed approach over the existing one depend on the position of the CD section of interfering flows. Consider again the example given in Figure~\ref{f:exmpl1}, where $ P_1 > P_2 $. Assuming that the path of $ f_1 $ and the length of the CD section are constant, i.e. $ |{\cal L}_1| = const, |{\cal L}^{CD}_{1,2}| = const $, from Equations~\ref{eq:sigmai}-\ref{eq:sigma} it straightforwardly follows that the improvements in the worst-case traversal time of $ f_2 $ are more influenced by the length of the pre-CD than the post-CD section. Thus, a necessary and sufficient condition for maximum improvements is: $ |{\cal L}^{pre-CD}_{1,2}| = |{\cal L}_1| - 1, |{\cal L}^{CD}_{1,2}| = 1 $ and $ |{\cal L}^{post-CD}_{1,2}| = 0 $.
	
	\begin{figure}[t!]
		\centering
		\includegraphics[width=0.9\columnwidth]{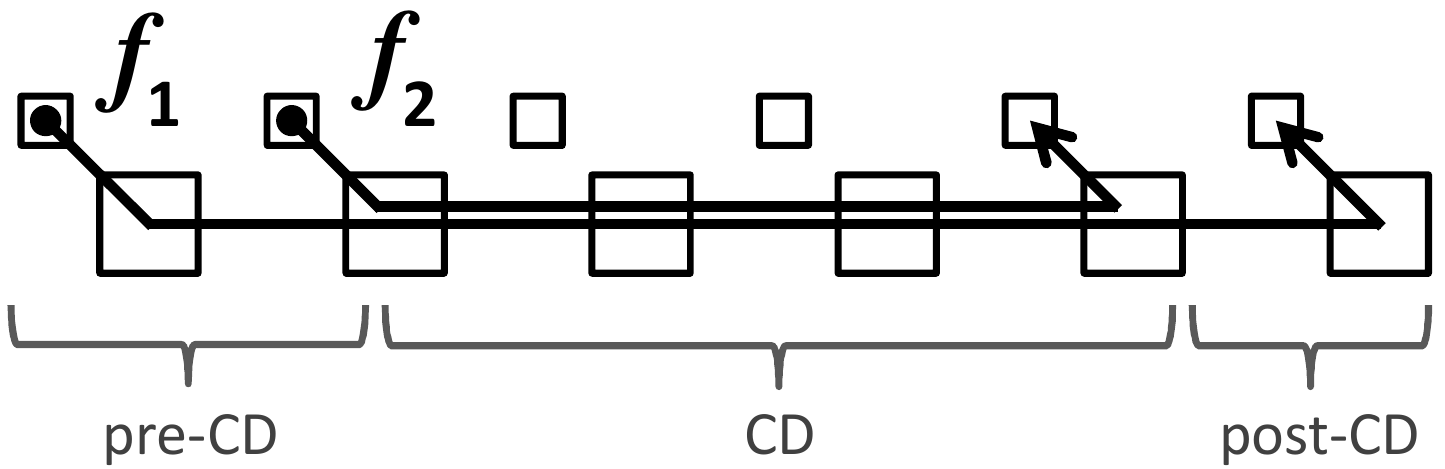}
  		\caption{\label{f:exmpl5} Example of traffic flows}
	\end{figure}
	
	To demonstrate that, let us analyse the example from Figure~\ref{f:exmpl6}, which differs from the example from Observation~1 (Figure~\ref{f:exmpl1}) only in the position of the CD section. If we compute the worst-case traversal times for both flows, assuming the same flow characteristics ($ size(f_1) = size(f_2) = 48B $, $ P_1 > P_2 $, $ T_1 = D_1 = T_2 = D_2 = 1000ns, J_1^R = J_2^R = 0 $), we reach the following values: $ R_1 = R^*_1 = 14ns $, $ R_2 = 20ns $ and $ R^*_2 = 12.5ns $. When compared with the example from Observation~1, it is visible that the improvements in the worst-case traversal time of $ f_2 $ increased from $ 6ns $ to $ 7.5ns $, or expressed relatively, from $ 30\% $ to $ 37.5\% $.
	
	\begin{figure}[t]
		\centering
		\includegraphics[width=0.9\columnwidth]{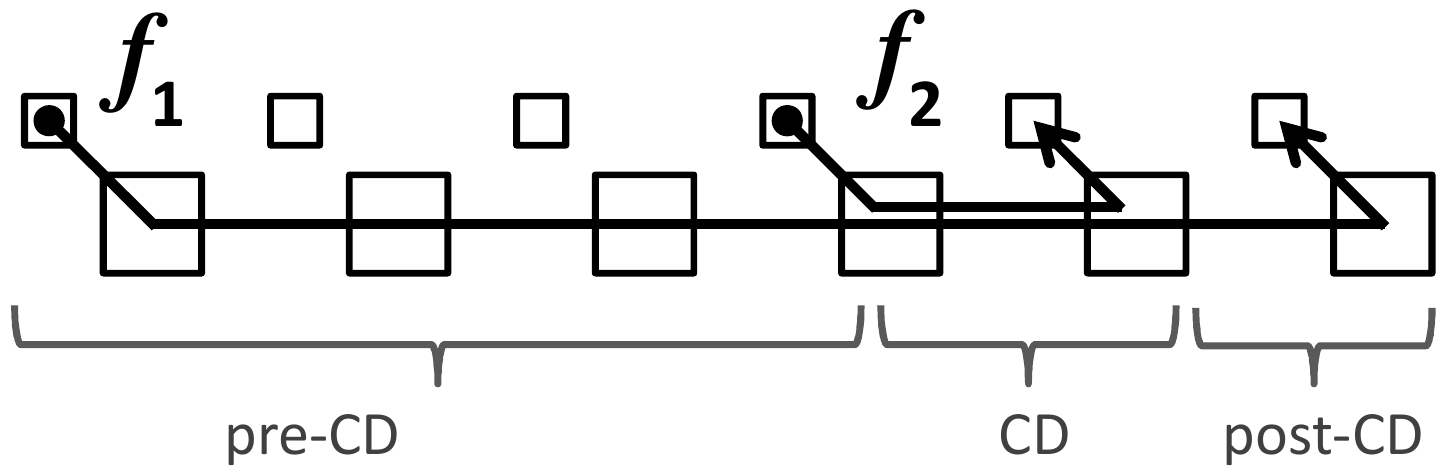}
  		\caption{\label{f:exmpl6} Example of traffic flows}
	\end{figure}
	
	{\bf Observation 4: }The improvements of the proposed approach over the existing one do not depend on flow-sizes. Indeed, from Equations~\ref{eq:sigmai}-\ref{eq:sigma} it is visible that only the paths, but not the sizes of the flows influence the improvements. Therefore, with the increase in flow sizes, the worst-case traversal times also grow, but the improvements remain the same in absolute values, and hence report decrease in relative values. If we compute the worst-case traversal times for the example given in Figure~\ref{f:exmpl1}, but this time with bigger flow sizes ($ size(f_1) = size(f_2) = 160B $, $ P_1 > P_2 $, $ T_1 = D_1 = T_2 = D_2 = 1000ns, J_1^R = J_2^R = 0 $), we get the following results: $ R_1 = R^*_1 = 17.5ns $, $ R_2 = 27ns $ and $ R^*_2 = 21ns $. Like in the equivalent case for $ size(f_1) = size(f_2) = 48B $, we again have $ R_2 - R^*_2 = 6ns $, however, the relative improvements drop from $ 30\% $ to $ 22.2\% $. This implies that, as flow sizes increase, the absolute improvements remain unaffected, however, the relative improvement decrease. Indeed, in a hypothetical case with flows of infinitely large sizes, the relative improvements asymptotically converge towards $ 0\% $.
	
\subsection{Numerical Example}
\label{sec:example}

	In order to get a better insight into how different parameters influence analysis improvements, we use a small-scale numerical example. We consider two contending flows $ f_1 $ and $ f_2 $ where $ P_1 > P_2 $ and ${\cal F}_D(f_2) = \{ f_1 \} $. We vary the length of the CD section, the sizes, and the paths of the flows. For each particular scenario we compute the worst-case traversal times of $ f_2 $ with both approaches, and observe the relative improvements of the proposed approach. Figures~\ref{f:exp1a}-\ref{f:exp1c} demonstrate the results. In each figure a lower surface covers a corner case where the length of the pre-CD section is equal to zero, i.e. $ |{\cal L}^{pre-CD}_{1,2}| = 0 $ and $ |{\cal L}^{post-CD}_{1,2}| = |{\cal L}_1| - |{\cal L}^{CD}_{1,2}| $, while the upper surface represents the opposite corner case, i.e. the length of the post-CD section is equal to zero, i.e. $ |{\cal L}^{pre-CD}_{1,2}| = |{\cal L}_1| - |{\cal L}^{CD}_{1,2}| $ and $ |{\cal L}^{post-CD}_{1,2}| = 0 $. The trends in Figures~\ref{f:exp1a}-\ref{f:exp1c} coincide with all the conclusions from Observations~1-4. Moreover, it is visible that the improvements are equal to zero in cases where the path of $ f_1 $ entirely belongs to the CD section, i.e. $ |{\cal L}^{CD}_{1,2}| = |{\cal L}_1| $, which has already been inferred in Observation~2.

\section{Evaluation}
\label{sec:experiments}
	
	\begin{table}[h]
		\centering
		\caption{ Analysis parameters}
		\begin{tabular}{|c|c|}
			\hline
			Platform size & \bf{8 $ \times $ 8} \\ 
			\hline
			Link width = flit size & \bf{16B} \\
			\hline
			NoC frequency & \bf{2GHz} \\
			\hline
			Router delay - $ d_r $ & \bf{3 cycles (1.5ns)} \\
			\hline
			Link delay - $ d_l $ & \bf{1 cycle (0.5ns)} \\
			\hline
			Flow periods & \bf{[1 - 10]ms}\\
			\hline
		\end{tabular}
		\label{tab:params}
	\end{table}
	
	In the previous section we have studied the improvements of the proposed analysis over the existing one on small-scale illustrative examples consisting of only two flows. In this section we perform a large-scale comprehensive experimental evaluation with the main objective to compare the two approaches, but this time assuming flow-sets with hundreds of flows. This will help us to investigate whether the improvement trends from a small-scale example also hold for large flow-sets, and to what extent. Subsequently, we can identify scenarios (flow characteristics) for which the proposed approach reports the biggest improvements. Analysis parameters are given in Table~\ref{tab:params}\footnote{A period of each flow is randomly generated, within the given limits. If a generated flow-set is not schedulable, then periods of all flows are uniformly increased (even beyond the limits) until the flow-set becomes schedulable.}.
	
	\begin{experiment} Improvements wrt flow sizes \end{experiment}
	
	\begin{figure}[t]
		\centering
  		\includegraphics[width=0.9\columnwidth]{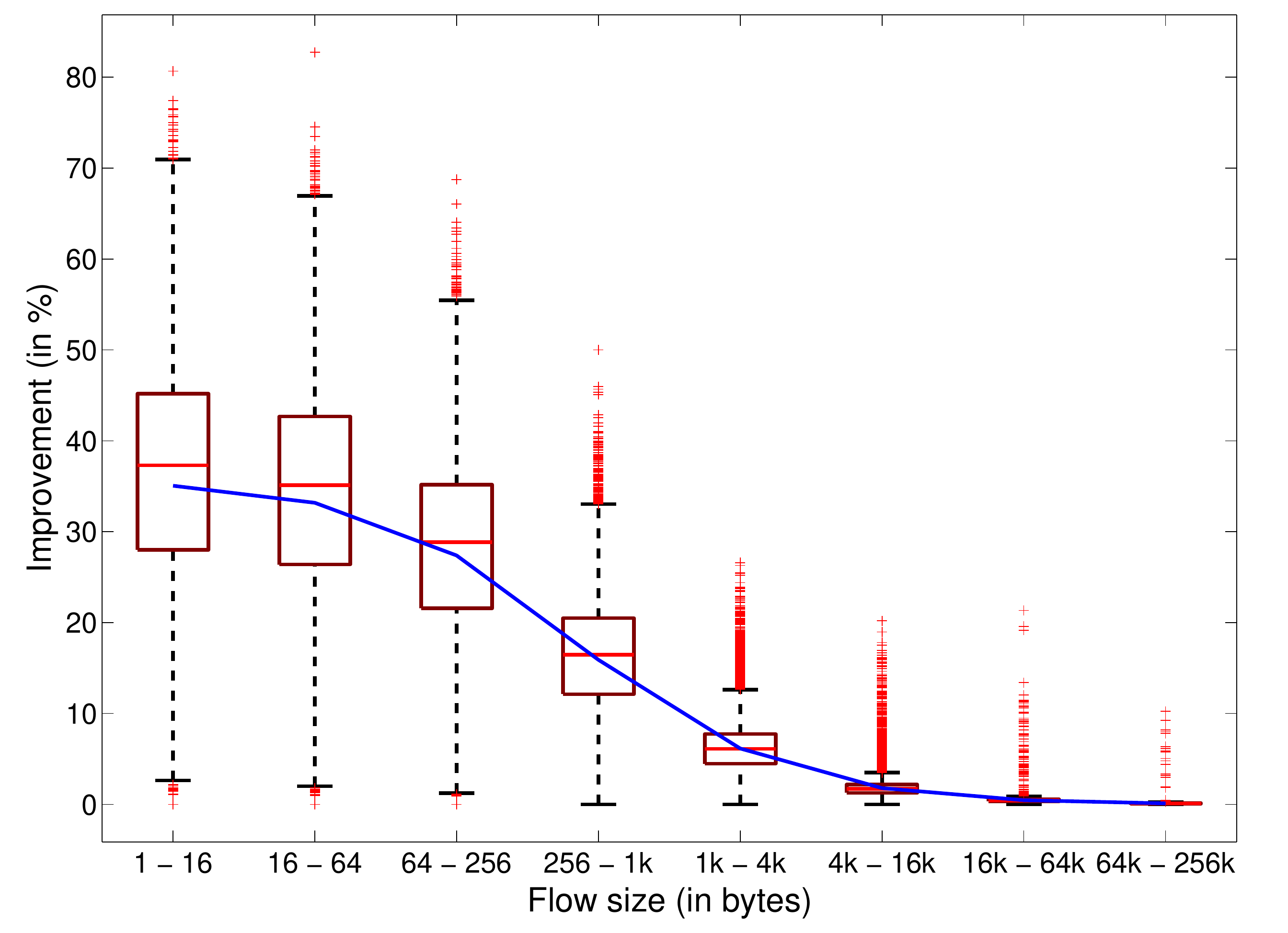}
  		\caption{\label{f:exp2} Improvements wrt flow sizes}
  	\end{figure}
	
	In this experiment we investigate how different flow sizes influence the improvements of the proposed analysis over the existing one. We generate several categories of flow-sets, with different flow size ranges: $ 1B - 16B, 16B - 64B, ..., 64kB - 256kB $. For each category we generate $ 100 $ flow-sets, where a flow-set consists of $ 200 $ flows. The size of each flow is randomly generated, within the limits of the respective flow-set category. Also, for each flow, the priority, the source core and the destination core are randomly generated, where flow paths comply with the XY routing policy. Subsequently, for each flow of the flow-set we compute the worst-case traversal time with both methods, and measure the improvements in relative terms.
	
	Figure~\ref{f:exp2} shows the results. Since the improvements do not depend on flow sizes (see Observation~4), the absolute improvements of the new approach over the existing one are constant, irrespective of the flow sizes. However, as the increase in the flow size causes a uniform increase in the worst-case traversal times obtained with both methods, the relative improvements of the new approach (y-axis) decrease as the flow sizes increase (x-axis). Another interesting finding is that, irrespective of the flow sizes, there are always flows for which the new approach does not yield better results. These are the highest-priority flows which do not suffer any interference, hence for them both methods return the same results.
	
	\begin{experiment} Improvements wrt paths sizes \end{experiment}
	
	\begin{figure}[t]
  		\centering
		\includegraphics[width=0.9\columnwidth]{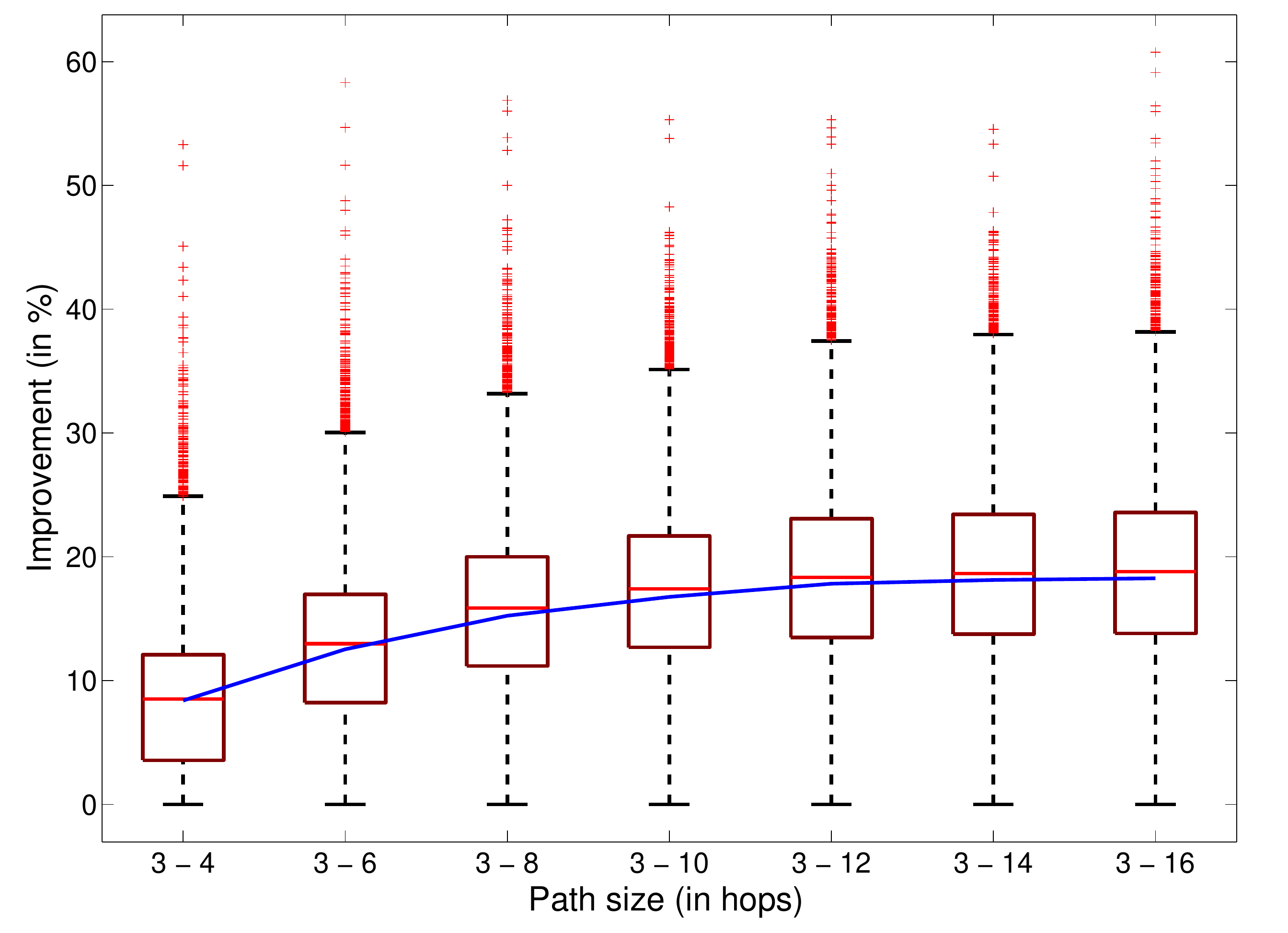}
		\caption{\label{f:exp3} Improvements wrt path sizes}
	\end{figure}
	
	In this experiment we vary the lengths of flow paths and observe their influence on the improvements. We again generate several categories of flow-sets, but this time with different lengths of flow paths: $ 3-4, 3-6, ..., 3-16 $. For each category $ 100 $ flow-sets are generated, where a flow-set consists of $ 200 $ flows. The priority, the source and destination cores are generated randomly for each flow, but in accordance with the constraint on the maximum path size, posed by the respective category to which the flow-set belongs. Each flow has a size which is randomly generated in the range $ [1B - 1kB] $. We compute the worst-case traversal times of all flows with both methods, and observe the improvements in relative terms.
	
	Figure~\ref{f:exp3} demonstrates that as the paths of the flows increase (x-axis), the relative improvements also increase (y-axis). The explanation is that short paths substantially decrease the number of contentions and interferences, hence decreasing the scenarios in which the new approach can cause improvements. In fact, even in scenarios where contentions do occur, due to short flow paths, CD sections cover large fractions of them, which has a significant impact on the improvements (Observation~2). Conversely, longer paths cause more interferences, but also longer pre-CD and post-CD sections. All these facts have a positive effect on the improvements (see Observation~2).
	
	\begin{experiment} Improvements wrt flow and path sizes \end{experiment}
	
	\begin{figure}[t]
		\centering
  		\includegraphics[width=0.9\columnwidth]{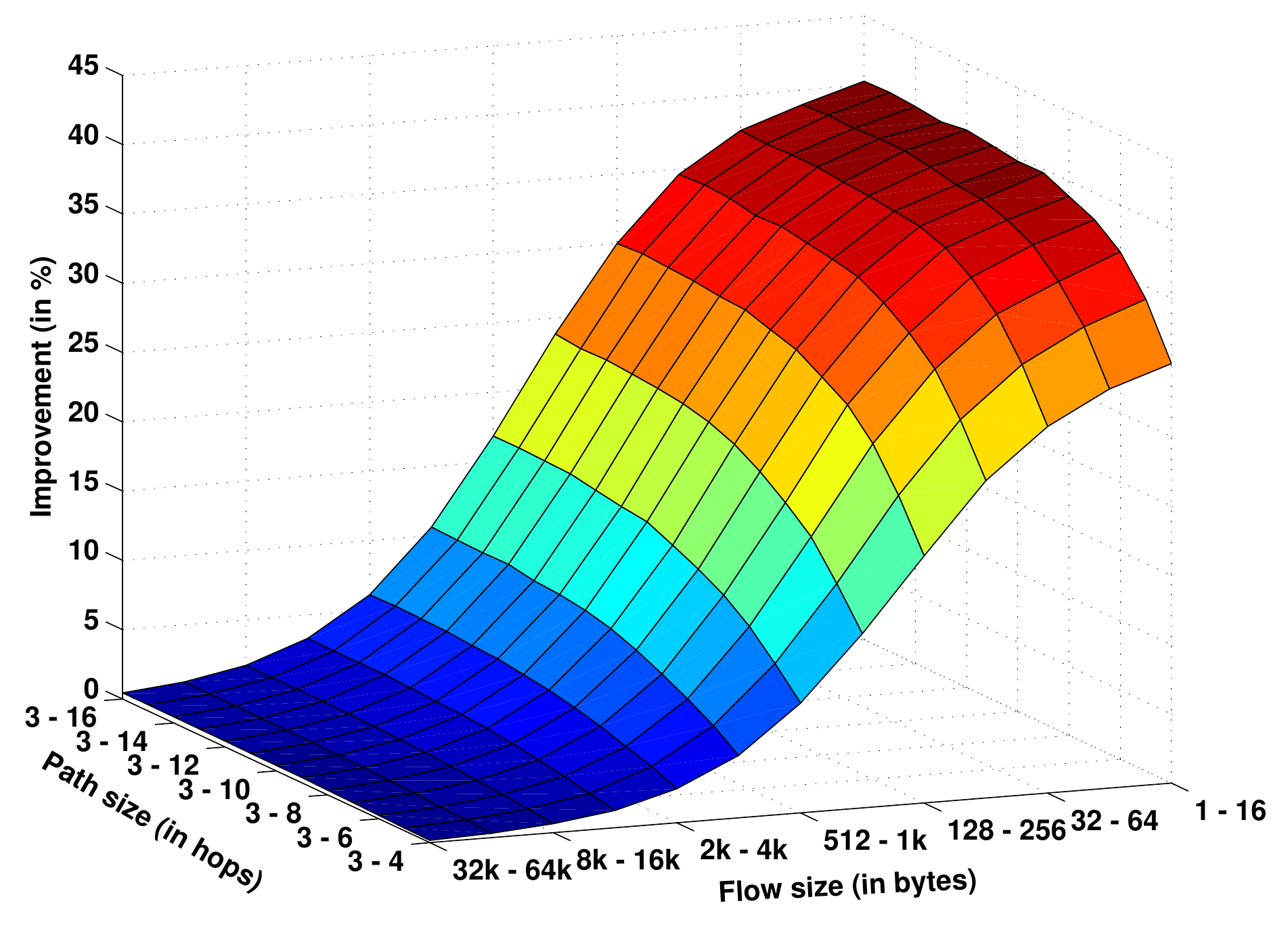}
  		\caption{\label{f:exp4} Improvements wrt flow and path sizes}
  	\end{figure}
	
	The objective of this experiment is to get a better insight into how both the aforementioned flow characteristics (the flow size and the path size) together influence the improvements of the new method over the existing one. We generate different flow-set categoties where we vary both parameters. Again, each category consists of $ 100 $ flow-sets, each with $ 200 $ flows with randomly generated priorities, source and destination cores. For each flow we compute the worst-case traversal times with both methods. Subsequently, for each category we observe the average improvements achieved with the new method, expressed in relative terms.
	
	Figure~\ref{f:exp4} shows the improvement trends (z-axis) associated with flow sizes (x-axis) and path sizes (y-axis). The improvement trends are identical to those from Experiments~1-2, inferring that the increase in the flow sizes and the decrease in the path sizes both have a negative effect on the relative improvements. This infers that small flows with long paths benefit the most from the proposed approach. Note, that a similar conclusion was reached for a small-scale numerical example with two flows (see Section~\ref{sec:example} and Figure~\ref{f:exp1}).
	
	\begin{experiment} Improvements wrt flow-set sizes \end{experiment}
	
	\begin{figure}[t]
  		\centering
  		\includegraphics[width=0.9\columnwidth]{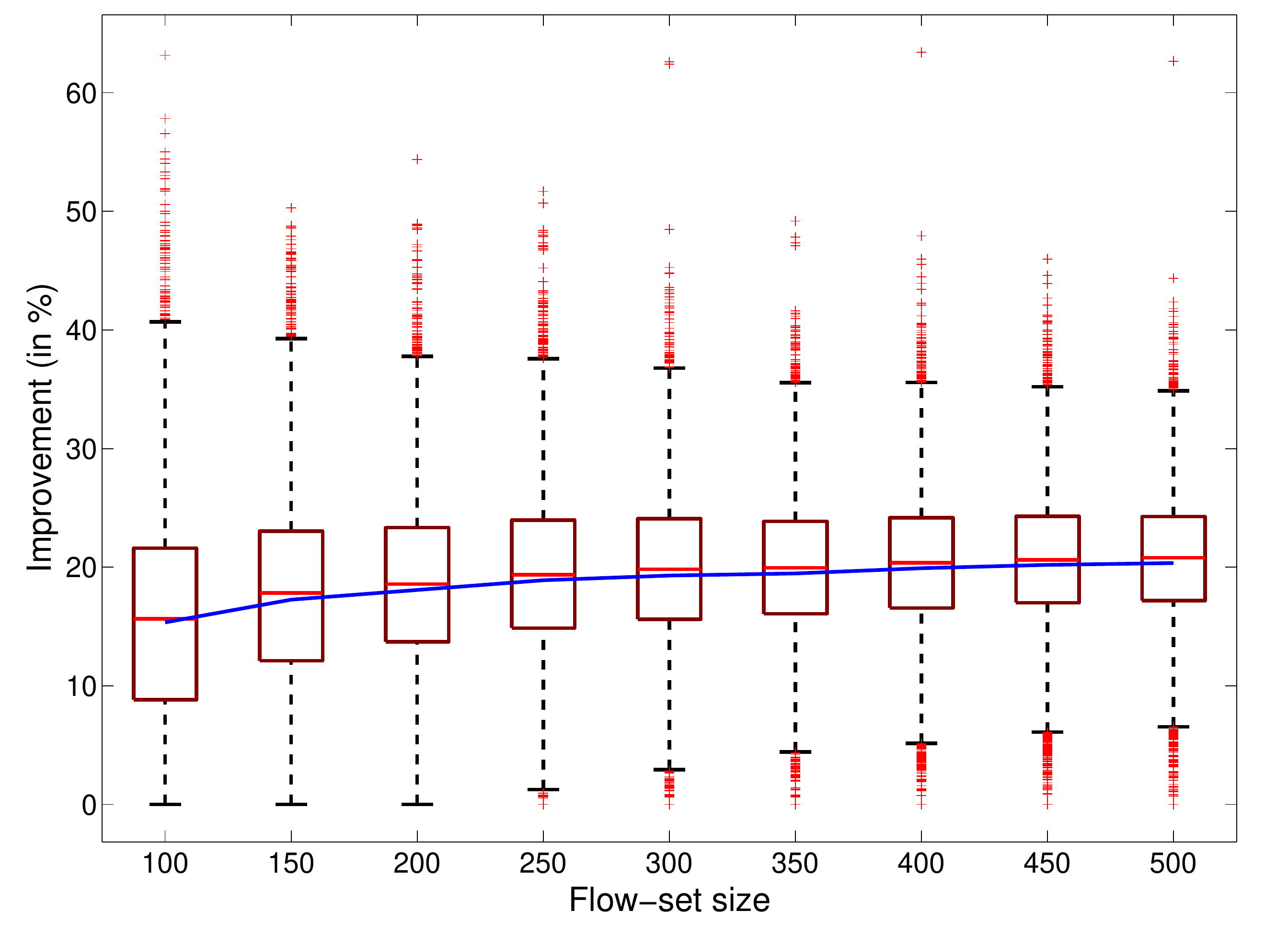}
		\caption{\label{f:exp5} Improvements wrt flow-set sizes}
	\end{figure}
	
	The emphasis of this experiment is on the flow-set size. In other words, we investigate how the improvement trends change with the number of flows constituting a flow-set. We generate several flow-set categories, where each category has the number of flows equal to one of these values $ 100, 150, ..., 500 $. For each category $ 100 $ flow-sets are generated, with random priorities, source and destination cores. Moreover, each flow has a size which is randomly generated in the range $ [1B - 1kB] $. Subsequently, for each flow we obtain the worst-case traversal times with both methods and compare the results in relative terms.
	
	Figure~\ref{f:exp5} demonstrates that as the flow-set size increases, so do the improvements. The explanation is that larger flow-sets have more substantial contentions, which favours the new approach. Yet, irrespective of the flow-set size, there always exist the highest-priority flows which do not suffer interference and hence for them no improvements can be achieved. However, as the flow-set size increases, the highest-priority flows constitute smaller and smaller fraction of the entire flow-set, hence for sets with more than $ 200 $ flows these cases are below the $ 25^{th} $ percentile, and are thus classified as outliers (depicted with red crosses in Figure~\ref{f:exp5}).
	
	\begin{experiment} Improvements wrt priorities \end{experiment}
	
	\begin{figure}[t]
		\centering
  		\includegraphics[width=0.9\columnwidth]{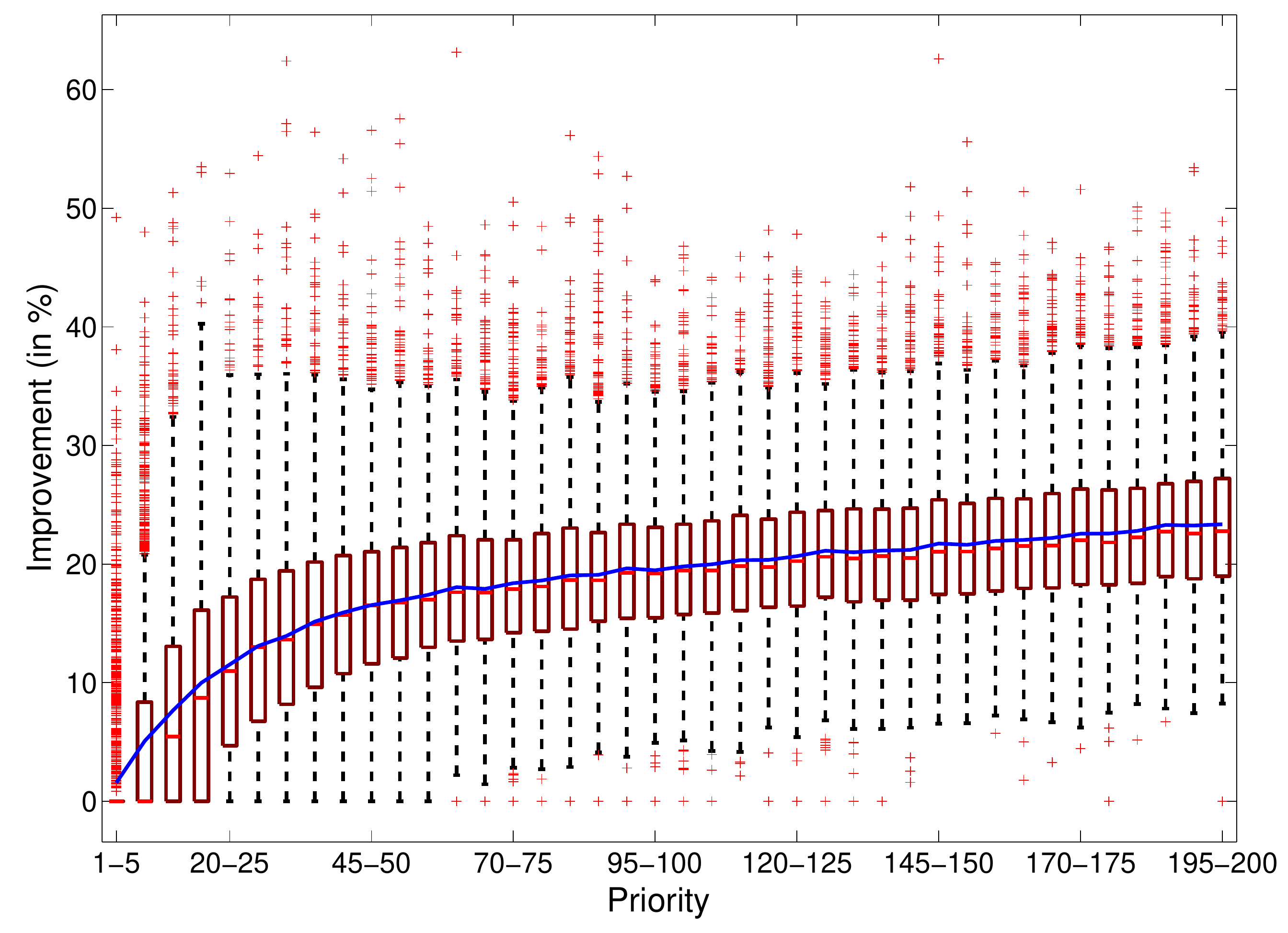}
  		\caption{\label{f:exp6} Improvements wrt priorities}
  	\end{figure}
	
	In this experiment we investigate how the improvement trends change with different flow priorities. We generated $ 100 $ flow-sets, each with $ 200 $ flows, where priorities, flow sizes, source and destination cores have been randomly generated. A size of each flow is a randomly generated value from the range $ [1B - 1kB] $. For each flow we compute the worst-case traversal times with both methods, and express the improvements achieved by the new approach, in relative terms.
	
	Figure~\ref{f:exp6} confirms that, as flow priorities decrease (bigger numbers on the x-axis), the relative improvements increase (y-axis). This confirms our initial assumption that the new approach does not produce significant improvements for the highest-priority flows, because these flows suffer very little interference (if at all). As flow priorities decrease, the interference that flows suffer becomes more substantial, which favours the new approach.
	
	\begin{experiment} Analysis tightness \end{experiment}
	
	\begin{figure}[t]
  		\centering
  		\includegraphics[width=0.9\columnwidth]{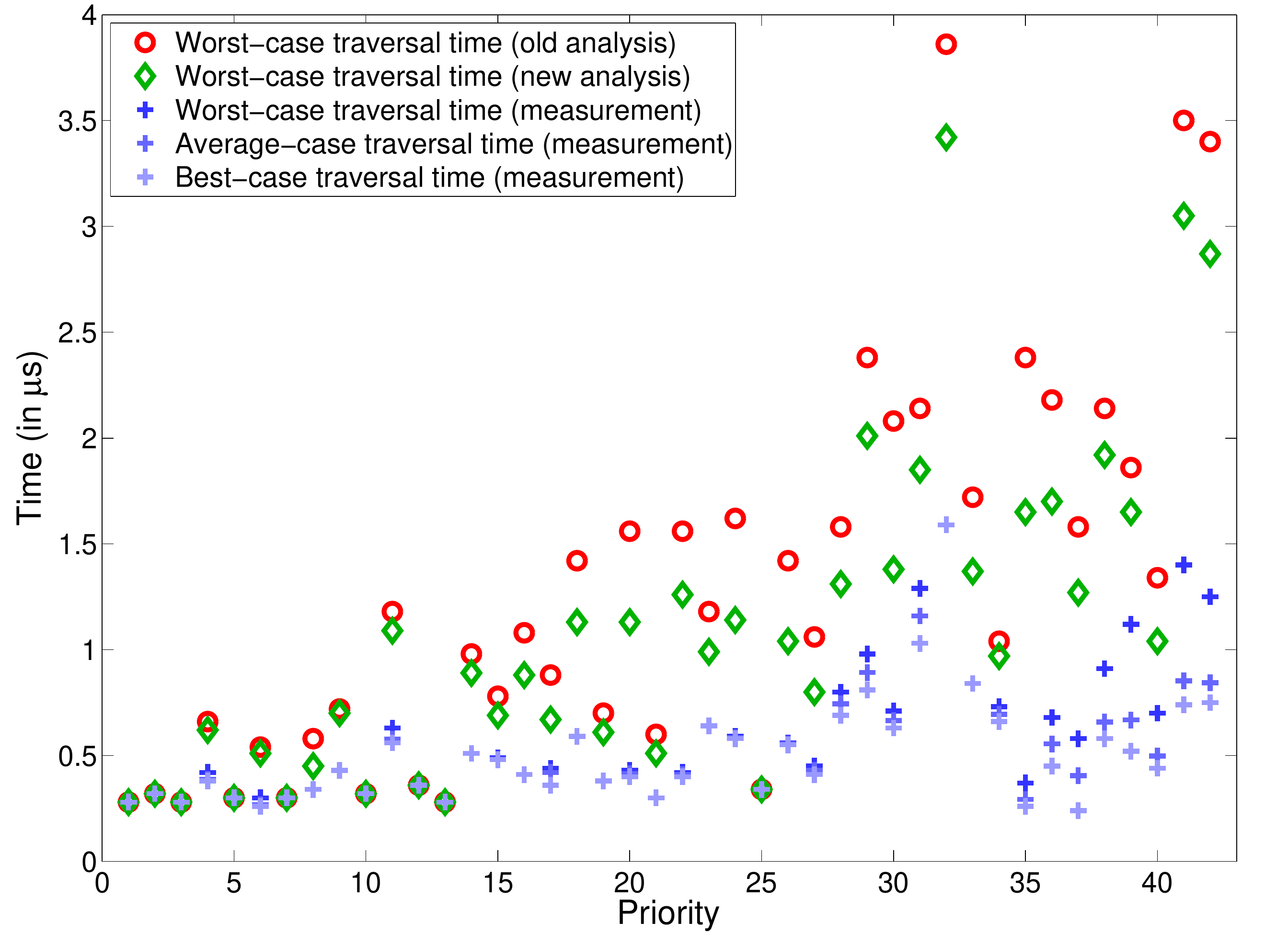}
		\caption{\label{f:exp7} Analysis tightness}
	\end{figure}
	
	The objective of this experiment is to investigate the tightness of the obtained upper-bounds on the worst-case traversal times. To achieve this, we generate a single flow-set consisting of $ 42 $ flows, and map it on a $ 6 \times 6 $ platform with randomly generated source and destination cores. Each flow has a payload in the range $ [2 - 48] $ flits, and one additional flit for a header. Moreover, each flow has a period in the range $ [0.5 - 9] ms $, and a unique priority. The NoC frequency is $ 100 $MHz. We compute the worst-case traversal times of all flows with both analyses, and also simulate the execution on a cycle-accurate simulator. The simulated time is $ 2 $ hyper-periods\footnote{A hyper-period is the least common multiplier of all flow-periods.}. Subsequently, we compare the analysis estimates with the values obtained via simulations, namely (i) the worst-case, (ii) the average-case and (iii) the best-case.
	
	Results are depicted in Figure~\ref{f:exp7}. It is visible that for high-priority flows both the analyses provide tight bounds. This is expected, given that these flows suffer very little interference, if at all. With the decrease in the priority (bigger numbers on the x-axis), the differences between the observed values and the analysis results become more noticeable (y-axis), because the analysis pessimism accumulates. Also, notice that the difference between the analysis estimates increases with the decrease in priorities. In fact, in some cases the proposed approach provides significantly tighter estimates, which entirely coincides with the conclusions from the previous experiments, and further motivates this work. However, the results also suggest that there is still room for improvement in the analysis, and we see this as a potential topic for future work.
	
\section{Conclusions and Future Work}

	The real-time communication analysis is the most suitable approach to determine whether a NoC-based many-core platform can accommodate the workload and always fulfil its timing requirements. In safety-critical and real-time computing areas, providing strong guarantees that all posed constraints will always be met, even in the worst-case conditions, is of paramount importance. Yet, one of the greatest challenges of analysis-based approaches is the pessimism, which can lead to a significant underutilisation of the platform and inefficient use of available resources.
	
	In this paper we have extended the existing real-time communication analysis~\cite{Zheng_Burns_08} for wormhole-switched priority-preemptive NoCs, and proposed a new, less pessimistic technique to compute tighter upper-bound estimates on the worst-case traversal times of traffic flows. Specifically, our approach reduces the pessimism of the direct interference that contending traffic flows cause to each other when competing for NoC resources. Subsequently, we have quantified the improvements (i.e. pessimism reduction, tightness) achieved with the new analysis, and also observed how these trends change with different flow parameters, namely, flow sizes, path sizes, flow-set sizes, priorities. The experiments demonstrate that the proposed approach yields significant improvements in almost all cases, while the greatest pessimism reductions are achieved in scenarios with large flow-sets, where flows have small sizes and traverse long paths. These traffic characteristics correspond to control core-to-core traffic as well as to read requests and write responses in core-to-memory traffic. Given that these traffic types constitute a significant fraction of the entire NoC traffic, the proposed analysis not only can help to exploit the platform more efficiently and decrease the resource overprovisioning, but also can render many flow-sets schedulable, even though the existing analysis classified them as unschedulable. Motivated with this fact, we plan to further explore the possibilities to improve the analysis, especially in the domain of indirect interferences, which have been, for clarity purposes, kept out of scope of this paper.

\balance
\bibliographystyle{abbrv}

\begin{thebibliography}{10}
\providecommand{\url}[1]{{#1}}
\providecommand{\urlprefix}{URL }
\expandafter\ifx\csname urlstyle\endcsname\relax
  \providecommand{\doi}[1]{DOI~\discretionary{}{}{}#1}\else
  \providecommand{\doi}{DOI~\discretionary{}{}{}\begingroup
  \urlstyle{rm}\Url}\fi

\bibitem{Audsley_BRTW_93}
N.~Audsley, A.~Burns, M.~Richardson, K.~Tindell, and A.~J. Wellings.
\newblock Applying new scheduling theory to static priority pre-emptive
  scheduling.
\newblock {\em Softw. Engin. J.}, 1993.

\bibitem{Benini_DeMicheli_02}
L.~Benini and G.~De~Micheli.
\newblock Networks on chips: a new soc paradigm.
\newblock {\em The Comp. J.}, 35(1):70 --78, jan 2002.

\bibitem{Dally_92}
W.~Dally.
\newblock Virtual-channel flow control.
\newblock {\em Trans. Parall. \& Distr. Syst.}, 3(2):194 --205, Mar 1992.

\bibitem{Dally_Seitz_87}
W.~Dally and C.~Seitz.
\newblock Deadlock-free message routing in multiprocessor interconnection
  networks.
\newblock {\em Trans. Computers}, 1987.

\bibitem{Dasari_NNP_13}
D.~Dasari, B.~Nikoli\'c, V.~Nelis, and S.~M. Petters.
\newblock Noc contention analysis using a branch and prune algorithm.
\newblock {\em Trans. Emb. Comput. Syst.}, 2013.

\bibitem{Ferrandiz_FF_09}
T.~Ferrandiz, F.~Frances, and C.~Fraboul.
\newblock A method of computation for worst-case delay analysis on spacewire
  networks.
\newblock In {\em SIES}, 2009.

\bibitem{Ferrandiz_FF_11}
T.~Ferrandiz, F.~Frances, and C.~Fraboul.
\newblock Using network calculus to compute end-to-end delays in spacewire
  networks.
\newblock {\em SIGBED Rev.}, 2011.

\bibitem{Goossens_DR_05}
K.~Goossens, J.~Dielissen, and A.~Radulescu.
\newblock Aethereal network on chip: concepts, architectures, and
  implementations.
\newblock {\em IEEE Design \& Test Computers}, 2005.

\bibitem{Hu_Marculescu_03}
J.~Hu and R.~Marculescu.
\newblock Energy-aware mapping for tile-based noc architectures under
  performance constraints.
\newblock In {\em 8th ASPDAC}, 2003.

\bibitem{SCC}
Intel.
\newblock {\em Single-Chip-Cloud Computer}.
\newblock
  \newline\url{www.intel.com/content/dam/www/public/us/en/documents/technology-briefs/intel-labs-single-chip-cloud-article.pdf}.

\bibitem{Kavaldjiev_Smit_03}
N.~K. Kavaldjiev and G.~J.~M. Smit.
\newblock A survey of efficient on-chip communications for soc.
\newblock In {\em 4th Symp. Emb. Syst.}, 2003.

\bibitem{Ni_McKinley_93}
L.~M. Ni and P.~K. McKinley.
\newblock A survey of wormhole routing techniques in direct networks.
\newblock {\em The Comp. J.}, 26, 1993.

\bibitem{Nikolic_APP_13}
B.~Nikoli\'c, H.~I. Ali, S.~M. Petters, and L.~M. Pinho.
\newblock Are virtual channels the bottleneck of priority-aware
  wormhole-switched noc-based many-cores?
\newblock In {\em 21st RTNS}, 2013.

\bibitem{Qian_LD_09}
Y.~Qian, Z.~Lu, and W.~Dou.
\newblock Analysis of worst-case delay bounds for best-effort communication in
  wormhole networks on chip.
\newblock In {\em NOCS}, 2009.

\bibitem{Zheng_Burns_08}
Z.~Shi and A.~Burns.
\newblock Real-time communication analysis for on-chip networks with wormhole
  switching.
\newblock In {\em NOCS}, 2008.

\bibitem{Song_KY_97}
H.~Song, B.~Kwon, and H.~Yoon.
\newblock Throttle and preempt: a new flow control for real-time communications
  in wormhole networks.
\newblock In {\em 1997 Int. Conf. Parall. Processing}, Aug 1997.

\bibitem{Tilera}
Tilera.
\newblock {\em TILE64\textsuperscript{\texttrademark} Processor}.
\newblock \newline\url{www.tilera.com/products/processors/TILEPro_Family}.

\end{thebibliography}


\providecommand{\noopsort}[1]{} \providecommand{\url}{\error{The bib files now
  require `url' package!}}

\end{document}